\documentclass[a4paper,onecolumn,11pt,accepted=2026-11-11]{quantumarticle}
\pdfoutput=1

\usepackage[utf8]{inputenc}
\usepackage[english]{babel}
\usepackage[T1]{fontenc}
\usepackage{amsmath}
\usepackage{amssymb}
\usepackage{amsthm}
\usepackage{mathtools}
\usepackage[numbers]{natbib}
\usepackage{hyperref}

\newtheorem{theorem}{Theorem}

\newtheorem{prop}[theorem]{Proposition}
\newtheorem{cor}[theorem]{Corollary}
\theoremstyle{definition}
\newtheorem{definition}[theorem]{Definition}
\newtheorem{remark}[theorem]{Remark}
\newtheorem{example}[theorem]{Example}

\newcommand{\microspace}{\mspace{0.5mu}}

\newcommand{\tr}{\operatorname{Tr}}
\newcommand{\pt}{\operatorname{T}}
\newcommand{\ppt}{\operatorname{PPT}}
\newcommand{\locc}{\operatorname{LOCC}}
\newcommand{\sep}{\operatorname{SEP}}
\newcommand{\opt}{\operatorname{opt}}

\renewcommand{\t}{{\scriptscriptstyle\mathsf{T}}}
\newcommand{\defeq}{\triangleq}

\newcommand{\ip}[2]{\langle #1 , #2\rangle}
\newcommand{\ket}[1]{\lvert\microspace #1 \microspace \rangle}
\newcommand{\bra}[1]{\langle\microspace #1 \microspace \rvert}
\newcommand{\ketbra}[2]{| #1 \rangle\langle #2 |}

\newcommand{\I}{\openone}
\newcommand{\1}{\mathbf{1}}

\newcommand{\setft}[1]{\mathrm{#1}}
\newcommand{\Density}{\setft{D}}
\newcommand{\Pos}{\setft{Pos}}
\newcommand{\DO}{\setft{DO}}
\newcommand{\Unitary}{\setft{U}}
\newcommand{\Herm}{\setft{Herm}}
\newcommand{\Lin}{\setft{L}}

\newcommand{\complex}{\mathbb{C}}
\newcommand{\real}{\mathbb{R}}

\newcommand{\X}{\mathcal{X}}
\newcommand{\Y}{\mathcal{Y}}
\newcommand{\E}{\mathcal{E}}
\renewcommand{\O}{\textup{O}}

\begin{document}

\title{Distinguishability of locally diagonal orthogonally invariant quantum
states}

\author{Nathaniel Johnston}
\email{njohnston@mta.ca}
\affiliation{Department of Mathematics \& Computer Science, Mount Allison
University, Canada}

\author{Vincent Russo}
\email{vincentrusso1@gmail.com}
\affiliation{Unitary Foundation}

\maketitle

\begin{abstract}
  We study the distinguishability of quantum states under local operations with
  classical communication (LOCC), separable, and positive-partial-transpose
  (PPT) measurements, focusing on \emph{locally diagonal orthogonally invariant}
  (LDOI) states---those invariant under local diagonal orthogonal twirling.
  This class includes many important families such as Werner states, isotropic
  states, X-states, and Dicke states. We show that optimal PPT and separable
  measurements for distinguishing LDOI states can always be taken to be LDOI,
  and the LOCC supremum can be approached by LDOI LOCC POVMs, enabling a
  dimensional reduction from $n^4$ to $O(n^2)$ in the associated optimization
  problems. We establish efficiently computable bounds on the distinguishability
  of orthonormal LDOI bases and prove that for a broad class of such
  bases---including all two-qubit cases---the LOCC supremum equals the PPT and
  separable optima. More generally, we show the gap
  between PPT and LOCC distinguishability is at most $(n-2)/(2n^2)$ for local
  dimension $n$.
\end{abstract}

\section{Introduction}

Quantum state distinguishability is a fundamental problem in quantum
information theory~\cite{bandyopadhyay2011more, bennett1999quantum,
ghosh2002local, horodecki2003local, virmani2001optimal, walgate2002nonlocality,
walgate2000local}: given an ensemble of quantum states $\{(p_i, \rho_i)\}$,
what is the maximum probability of correctly identifying which state was
selected? The answer depends critically on what class of measurements is
allowed. Global measurements achieve the optimal success probability given
by the Holevo-Helstrom bound~\cite{helstrom1969quantum, holevo1973statistical} when
there are two states to distinguish, and more generally by the
Holevo-Yuen-Kennedy-Lax optimality
conditions~\cite{holevo1973statistical, yuen1975optimum}. However, in many
scenarios---particularly in distributed
quantum information processing---the parties performing the measurement may be
restricted to \emph{local operations with classical communication}
(LOCC)~\cite{chitambar2014everything}, \emph{separable
measurements}~\cite{duan2009distinguishability, bandyopadhyay2015limitations},
or \emph{positive-partial-transpose} (PPT)
measurements~\cite{cosentino2013positive, yu2014distinguishability}. These measurement classes
form a hierarchy $\mathrm{LOCC} \subseteq \mathrm{SEP} \subseteq \mathrm{PPT}
\subseteq \mathrm{Global}$, where each inclusion can be strict.

Determining the optimal distinguishability under these restricted measurement
classes is generally a difficult problem. While semidefinite programming
provides exact characterizations for both global and PPT
distinguishability~\cite{cosentino2013positive, bandyopadhyay2015limitations},
and hierarchies of SDPs can approximate separable
distinguishability~\cite{bandyopadhyay2015limitations}, the resulting
optimization problems can be computationally intensive, and exact formulas are
known only for specific families of states~\cite{bandyopadhyay2015limitations,
yu2012four}. Understanding when
these different measurement classes achieve the same optimal value, and
quantifying the gaps between them, remains an active area of research.

In this work, we focus on the class of \emph{locally diagonal orthogonally
invariant} (LDOI) states---those that are invariant under the local diagonal
orthogonal twirl map~\cite{chruscinski2006class, johnston2019pairwise, nechita2021graphical, nechita2021diagonal}. This class encompasses many important families of quantum
states including Werner states~\cite{werner1989quantum}, isotropic
states~\cite{horodecki1999reduction}, X-states~\cite{quesada2012quantum},
mixtures of Dicke states~\cite{yu2016separability, tura2018separability}, and
various entanglement witnesses such as the Choi witness~\cite{ha2011one}. Our
first main result shows that for any ensemble of LDOI states, optimal PPT and
separable measurements can always be taken to be LDOI, and the LOCC supremum
can be approached by LDOI measurements (Theorem~\ref{thm:ldoi_sdp_simplifies}),
enabling a significant dimensional reduction in the associated optimization
problems.

The technique of twirling to project onto an invariant subspace is
well-established in quantum information theory---it is used, for example, in the
study of Werner and isotropic state
entanglement~\cite{werner1989quantum,horodecki1999reduction}. Our contribution
is in identifying that the LDOI class is simultaneously broad enough to capture
many important families of states and structured enough for the resulting
optimization problems to remain tractable.

We establish efficiently computable upper and lower bounds on the success
probability of distinguishing orthonormal LDOI bases. We focus on uniform
ensembles over orthonormal bases because basis discrimination with a uniform
prior is a fundamental primitive in quantum information: it arises naturally in
superdense coding~\cite{bennett1992communication}, random access
codes~\cite{ambainis1999dense}, and entanglement-assisted
communication~\cite{bennett1999quantum}, and every orthonormal basis composed of
LDOI states---Bell bases, generalized Bell bases, and Dicke-type
bases---falls within this framework. These bases can be
parameterized by a pair of matrices $(U,A)$ where $U$ is unitary and $A$
satisfies certain orthogonality constraints (Definition~\ref{defn:LDOI_basis}).
Our lower bound for LOCC measurements can be computed in polynomial time via
the Hungarian algorithm for the Assignment
Problem~\cite{kuhn1955hungarian, munkres1957algorithms, tomizawa1971some,
edmonds1972theoretical},
while our upper bound for PPT measurements requires solving a semidefinite
program that is significantly smaller than the general PPT distinguishability
SDP---exploiting the block structure of LDOI states to reduce the search space
from dimension $n^4$ to $3n^2$.

For a broad class of LDOI bases---including all two-qubit LDOI bases---our
upper and lower bounds coincide, yielding an exact closed-form formula showing
that $\opt_{\locc} = \opt_{\sep} =
\opt_{\ppt}$ (Corollaries~\ref{cor:large_U}
and~\ref{cor:ppt_dist_2x2}). This generalizes previous results on parametrized
Bell states~\cite{bandyopadhyay2013tight, bandyopadhyay2015limitations} to the
entire class of orthonormal LDOI bases in dimension two. More generally, we
prove that the gap between PPT and LOCC distinguishability for any LDOI basis
is at most $(n-2)/(2n^2)$, which vanishes as the dimension grows and achieves
its maximum value of $1/16$ when $n=4$
(Corollary~\ref{cor:close_LDOI_LOCC_PPT}). The LOCC lower bound evaluates to
exactly $1/2 - (n-2)/(2n^2)$ for the Fourier basis.

The remainder of this paper is organized as follows. After establishing
notation and preliminaries, Section~\ref{sec:LDOI_dist} introduces LDOI states
and bases and establishes their block structure.
Section~\ref{sec:LDOI_distinguishability} proves our main distinguishability
bounds. Section~\ref{sec:conclusion} concludes with a discussion of open
problems and future directions.

\subsection*{Preliminaries}
We assume the reader is familiar with the notions of quantum information as
contained in~\cite{nielsen2002quantum} and~\cite{wilde2013quantum}. We will make
use of the notation and terminology found in~\cite{watrous2018theory}. Let $\X =
\complex^n$ denote a finite-dimensional complex Euclidean space with a fixed
standard basis $\{\ket{1}, \ldots, \ket{n} \}$ for some positive integer $n$. We
use sets $\Lin(\X), \Pos(\X), \Herm(\X), \Density(\X)$, and $\Unitary(\X)$ to
represent the set of linear operators, positive semidefinite operators,
Hermitian operators, density operators, and unitary operators acting on the
space $\X$. For a vector $\ket{\phi}$ we adopt the convention of representing
the corresponding density operator as $\phi = \ketbra{\phi}{\phi}$. We denote
the transpose mapping $T : \Lin(\X) \rightarrow \Lin(\X)$ as the positive
mapping defined as $T(X) = X^{\t}$ for all $X \in \Lin(\X)$. The partial
transpose is a mapping on $\X \otimes \Y$ defined as
\begin{equation}
    \pt_{\X} = T \otimes \I_{\Y}.
\end{equation}
We write $\I_n$ for the $n \times n$ identity matrix and $\1_n$ for the $n \times n$ all-ones matrix (whose entries are all equal to $1$). We denote by $S_n$ the symmetric group on $\{1,\ldots,n\}$, i.e., the set of all permutations of $n$ elements.

\subsection*{Quantum state distinguishability}

We now formalize the state discrimination problem introduced above. We consider
bipartite quantum states shared between two parties, Alice and Bob, each holding
a subsystem. An \emph{ensemble} is a collection of pairs
\begin{equation}
    \E = \left\{
        \left(p_1, \rho_1\right), \ldots,
        \left(p_N, \rho_N\right)
        \right\},
\end{equation}
where each $\rho_i \in \Density(\X \otimes \Y)$ is a density operator
representing a quantum state shared between Alice and Bob, and $p_i > 0$ is the
prior probability with which $\rho_i$ is selected, satisfying $\sum_{i=1}^N p_i
= 1$. In the state discrimination task, a random index $k \in \{1, \ldots, N\}$
is selected according to the distribution $\{p_i\}$, Alice and Bob are provided
with the state $\rho_k$, and their goal is to determine $k$ by performing
a measurement on their respective portions of the shared state.
In most of this paper, the ensemble states are pure states $\rho_k =
\ketbra{\phi_k}{\phi_k}$ for unit vectors $\ket{\phi_k} \in \X \otimes \Y$,
though the definitions and Theorem~\ref{thm:ldoi_sdp_simplifies} apply to
mixed-state ensembles as well.

The optimal success probability depends critically on what class of measurements
is permitted. If Alice and Bob can perform arbitrary joint (global)
measurements, the maximum success probability is given by the Holevo-Helstrom
theorem when $N = 2$, and by the Holevo-Yuen-Kennedy-Lax optimality
conditions~\cite{holevo1973statistical, yuen1975optimum} for arbitrary $N$; in
either case it is the optimal value of a semidefinite
program~\cite{watrous2018theory}. However, in many practical
scenarios---particularly in distributed
quantum information processing---Alice and Bob are spatially separated and can
only perform \emph{local operations with classical communication}
(LOCC)~\cite{chitambar2014everything}. This restriction can strictly reduce the
distinguishability of certain ensembles. Two natural relaxations of LOCC
measurements are \emph{separable measurements}
(SEP)~\cite{duan2009distinguishability, bandyopadhyay2013tight,
bandyopadhyay2015limitations}, whose POVM elements lie in the cone of separable
operators (finite sums of positive product operators), and
\emph{positive-partial-transpose measurements}
(PPT)~\cite{cosentino2013positive, cosentino2013small, yu2014distinguishability,
bandyopadhyay2015limitations}, which require that the measurement operators
remain positive under partial transpose. These classes form a strict hierarchy
\begin{equation}
    \text{LOCC} \subseteq \text{SEP} \subseteq \text{PPT} \subseteq \text{Global}.
\end{equation}

Computing the LOCC supremum is generally intractable, as LOCC protocols can
involve arbitrary rounds of classical communication. Separable
measurements are characterized by a difficult convex optimization problem.
However, PPT measurements admit a tractable semidefinite programming
characterization, which we now describe. 

For an ensemble $\E = \{(p_i, \rho_i)\}_{i=1}^N$, the optimal success
probability under PPT measurements can be characterized via semidefinite
programming~\cite{cosentino2013positive, bandyopadhyay2015limitations}. A
measurement is described by a positive operator-valued measure (POVM) $\{P_k\}$,
where $P_k$ represents the measurement operator associated with outcome $k$. The
success probability is $\sum_{k=1}^N p_k \ip{P_k}{\rho_k}$. For PPT measurements, we require each $P_k$ and its
partial transpose $\pt_{\X}(P_k)$ to be positive semidefinite. The primal and
dual semidefinite programs characterizing $\opt_{\ppt}(\E)$ are
\begin{center}
  \begin{minipage}{5in}
    \centerline{\underline{Primal problem}}\vspace{-3mm}
    \begin{equation} \label{eq:ppt-primal}
      \begin{aligned}
        \text{maximize:} \quad & \sum_{k=1}^N p_k\ip{P_k}{\rho_k} \\
            \text{subject to:} \quad & \sum_{k=1}^N P_k = \I_{\X \otimes \Y}, \\
                              & P_1, \ldots, P_N \in \Pos(\X \otimes \Y), \\
                              & \pt_{\X}(P_1), \ldots, \pt_{\X}(P_N) \in \Pos(\X \otimes \Y).
    \end{aligned}
    \end{equation}
  \end{minipage}\vspace*{0.1in}
  \begin{minipage}{5in}
    \centerline{\underline{Dual problem}}\vspace{-3mm}
    \begin{equation} \label{eq:ppt-dual}
      \begin{aligned}
        \text{minimize:} \quad & \tr(H) \\
            \text{subject to:} \quad & H - p_k \rho_k - \pt_{\X}(Q_k) \in \Pos(\X \otimes \Y),
                              \quad k = 1,\ldots, N, \\
                              & H \in \Herm(\X \otimes \Y), \\
                              & Q_1, \ldots, Q_N \in \Pos(\X \otimes \Y).
    \end{aligned}
    \end{equation}
  \end{minipage}
\end{center}
Setting $Q_k = 0$ in Equation~\eqref{eq:ppt-dual} and taking the partial transpose (which preserves feasibility because $\pt_{\X}$ is self-adjoint) yields the simplified dual
\begin{center}
  \begin{minipage}{5in}
    \centerline{\underline{Dual problem (upper bound)}}\vspace{-3mm}
    \begin{equation} \label{eq:ppt-dual-constrained}
      \begin{aligned}
        \text{minimize:} \quad & \tr(H) \\
            \text{subject to:} \quad & H - p_k\pt_{\X}(\rho_k) \in \Pos(\X \otimes \Y),
                              \quad k = 1,\ldots, N, \\
                              & H \in \Herm(\X \otimes \Y).
    \end{aligned}
    \end{equation}
  \end{minipage}
\end{center}
The primal problem directly optimizes over PPT measurements, while the dual
provides a certificate for upper bounds on the optimal value. The constrained
dual in Equation~\eqref{eq:ppt-dual-constrained} therefore follows from fixing $Q_k = 0$ and using the self-adjointness of the partial transpose, yielding a simpler optimization problem at the cost of producing only an upper bound. If we define $\alpha$, $\beta$, and
$\beta^{\prime}$ to represent the optimal values of the primal
(Equation~\eqref{eq:ppt-primal}), the dual (Equation~\eqref{eq:ppt-dual}), and
the constrained dual (Equation~\eqref{eq:ppt-dual-constrained}), respectively,
then by weak duality we have $\alpha \leq \beta \leq \beta^{\prime}$. When the
primal and dual achieve the same optimal value ($\alpha = \beta$), strong
duality holds and the optimal PPT distinguishability is exactly determined.

Throughout this paper, we denote ensembles of quantum states explicitly as
$\E = \{(p_i, \rho_i)\}_{i=1}^N$ where $p_i$ are the prior probabilities and
$\rho_i$ are the density matrices. For an orthonormal basis $\eta =
\{\ket{\phi_1}, \ldots, \ket{\phi_N}\}$, we write $\E = \{(1/N,
\ketbra{\phi_i}{\phi_i})\}_{i=1}^N$ to denote the ensemble with uniform prior
distribution. We use the notation $\opt_{\ppt}(\E)$ and
$\opt_{\sep}(\E)$ for the optimal success probabilities under PPT and
separable measurements, and $\opt_{\locc}(\E)$ for the supremum over LOCC
measurements (since LOCC is not topologically closed, the supremum may not be
attained).

\section{Locally diagonal orthogonally invariant states and bases}
\label{sec:LDOI_dist}

Let $n \geq 2$ be an integer (the case $n=1$ is trivial) and let $\X = \Y = \complex^n$ be complex
Euclidean spaces. Let $\DO(\X)$ denote the set of diagonal matrices with
diagonal entries equal to $\pm 1$ acting on $\X$. (These are sometimes called
diagonal sign matrices or diagonal orthogonal matrices, as each element $O \in
\DO(\X)$ satisfies $O^2 = \I$ and $O = O^*$.) Consider the \emph{local
diagonal orthogonal twirl} map
$\Phi_{\O} : \Lin(\X \otimes \Y) \rightarrow \Lin(\X \otimes \Y)$
defined by
\begin{equation}
  \Phi_{\O}(A) \defeq \frac{1}{2^n}\sum_{O \in \DO(\X)} (O \otimes O)A(O \otimes O).
\end{equation}
Since $|\DO(\X)| = 2^n$, this map can be thought of as an average of $A$ over
symmetric local diagonal orthogonal conjugations. We begin with a few notes
about $\Phi_\O$, all of which are straightforward to verify. Here $\ppt(\X : \Y)$ denotes the cone of positive semidefinite operators whose partial transpose on $\X$ is also positive semidefinite, and $\sep(\X : \Y)$ denotes the cone of separable (convex combinations of product) operators on $\X \otimes \Y$.

\begin{itemize}
  \item $\Phi_{\O}$ is a quantum channel (i.e., completely positive and trace-preserving).

  \item $\Phi_{\O}$ is self-dual (with respect to the usual trace inner product): $\Phi_{\O}^* = \Phi_{\O}$.

  \item If $A \in \ppt(\X : \Y)$ then $\Phi_{\O}(A) \in \ppt(\X : \Y)$.

  \item If $A \in \sep(\X : \Y)$ then $\Phi_{\O}(A) \in \sep(\X : \Y)$.
\end{itemize}

The twirl map $\Phi_\O$ has a simple form in terms of the computational
basis (see \cite{johnston2019pairwise, nechita2021diagonal}, for example): if
the entry of $A \in \Lin(\X \otimes \Y)$ corresponding to the basis matrix
$\ketbra{i}{j} \otimes \ketbra{k}{\ell}$ is denoted by $a_{i,j;k,\ell}$ then
\begin{equation}
  \Phi_\O(A) =
  \sum_{i,j=1}^n a_{i,i;j,j}\ketbra{i}{i} \otimes \ketbra{j}{j} +
  \sum_{\substack{i,j=1 \\ i \neq j}}^n a_{i,j;i,j}\ketbra{i}{j} \otimes \ketbra{i}{j} +
  \sum_{\substack{i,j=1 \\ i \neq j}}^n a_{i,j;j,i}\ketbra{i}{j} \otimes \ketbra{j}{i}.
\end{equation}
We refer to $\Phi_{\O}$ as the \emph{local diagonal orthogonal twirl} (LDOT) map; it is the orthogonal projection onto the set of
matrices whose only non-zero entries are located in the $\ketbra{i}{i} \otimes
\ketbra{j}{j}$, $\ketbra{i}{j} \otimes \ketbra{i}{j}$, or $\ketbra{i}{j}
\otimes \ketbra{j}{i}$ positions for some $1 \leq i,j \leq n$.

A matrix $A \in \Lin(\X \otimes \Y)$ is called \emph{locally diagonal orthogonally
invariant (LDOI)} if $\Phi_\O(A) = A$. Many important families of
quantum states and operators studied in quantum information theory are LDOI,
including Werner states, isotropic states, X-states, mixtures of Dicke states,
and various entanglement witnesses~\cite{werner1989quantum,
horodecki1999reduction, quesada2012quantum, yu2016separability,
tura2018separability, ha2011one}. Detailed examples of these families are
provided in Section~\ref{subsec:LDOI_examples}.

To illustrate the block structure of LDOI matrices, consider first the case
$n=2$. The two-qubit LDOI states, known as X-states~\cite{quesada2012quantum},
have matrix representation of the form
\begin{equation}
  \rho = \begin{bmatrix}
    * & \cdot & \cdot & * \\
    \cdot & * & * & \cdot \\
    \cdot & * & * & \cdot \\
    * & \cdot & \cdot & *
  \end{bmatrix} \in \Density(\X \otimes \Y),
\end{equation}
where dots denote entries equal to $0$ and asterisks ($*$) denote potentially
non-zero entries.

In all dimensions, LDOI states have a very simple block matrix structure much
like that of the X-states: we can write them as a direct sum of a $n \times n$
matrix and $n(n-1)/2$ different $2 \times 2$ matrices as follows:
\begin{equation}
  \begin{aligned}
    \rho & = \left(\sum_{i,j=1}^n \rho_{i,j;i,j}\ketbra{i}{j} \otimes \ketbra{i}{j}\right) \\
    & \quad + \sum_{\substack{i,j=1 \\ i > j}}^n \big( \rho_{i,i;j,j}\ketbra{i}{i} \otimes \ketbra{j}{j} + \rho_{i,j;j,i}\ketbra{i}{j} \otimes \ketbra{j}{i} + \rho_{j,i;i,j}\ketbra{j}{i} \otimes \ketbra{i}{j} + \rho_{j,j;i,i}\ketbra{j}{j} \otimes \ketbra{i}{i} \big).
  \end{aligned}
\end{equation}
For example, if $n = 3$ then the LDOI states $\rho \in \Density(\X \otimes \Y)$
are the ones with a matrix representation of the form
\begin{equation}
  \rho = \left[\begin{array}{@{}ccc|ccc|ccc@{}}
    * &  \cdot &  \cdot &  \cdot & * &  \cdot &  \cdot &  \cdot & * \\
    \cdot &  * &  \cdot & * &  \cdot &  \cdot &  \cdot &  \cdot &  \cdot \\
    \cdot &  \cdot &  * &  \cdot &  \cdot &  \cdot & * &  \cdot &  \cdot \\\hline
    \cdot & * &  \cdot &  * &  \cdot &  \cdot &  \cdot &  \cdot &  \cdot \\
    * &  \cdot &  \cdot &  \cdot &  * &  \cdot &  \cdot &  \cdot & * \\
    \cdot &  \cdot &  \cdot &  \cdot &  \cdot &  * &  \cdot & * &  \cdot \\\hline
    \cdot &  \cdot & * &  \cdot &  \cdot &  \cdot &  * &  \cdot &  \cdot \\
    \cdot &  \cdot &  \cdot &  \cdot &  \cdot & * &  \cdot &  * &  \cdot \\
    * &  \cdot &  \cdot &  \cdot & * &  \cdot &  \cdot &  \cdot &  *
  \end{array}\right],
\end{equation}
where we can see that such a $\rho$ is a direct sum of the $3 \times 3$
submatrix defined by rows $1$, $5$, and $9$, as well as the $2 \times 2$
submatrices defined by rows $2$ and $4$, $3$ and $7$, and $6$ and $8$.

\subsection{Structure and parameterization of LDOI matrices}
\label{subsec:LDOI_structure}

The block structure exhibited above has a remarkably simple algebraic
characterization: every LDOI matrix can be uniquely parameterized by a triple of
$n \times n$ matrices sharing a common diagonal. This result, established
in~\cite{nechita2021diagonal}, significantly reduces the complexity of working
with LDOI matrices and provides both theoretical and computational advantages.

\begin{prop}[\cite{nechita2021diagonal}]\label{prop:triple_param}
  The space of LDOI matrices is in bijection with triples $(A,B,C)$ of $n
  \times n$ complex matrices satisfying $\text{diag}(A) = \text{diag}(B) =
  \text{diag}(C)$, where $\text{diag}(M)$ denotes the vector of diagonal entries
  of $M$. Specifically, an LDOI matrix $X \in \Lin(\X \otimes \Y)$ corresponds
  to the triple $(A,B,C)$ via
  \begin{equation}\label{eq:triple_to_ldoi}
    X = \sum_{i,j=1}^n a_{ij} \ketbra{i}{i} \otimes \ketbra{j}{j} +
    \sum_{\substack{i,j=1 \\ i\neq j}}^n b_{ij} \ketbra{i}{j} \otimes \ketbra{i}{j} +
    \sum_{\substack{i,j = 1 \\ i \neq j}}^n c_{ij} \ketbra{i}{j} \otimes \ketbra{j}{i}.
  \end{equation}
\end{prop}

This parameterization reveals that the LDOI subspace has dimension $3n^2 - 2n$
over the reals (for Hermitian matrices), compared to dimension $n^4$ for general
$n \times n$ bipartite matrices. The shared diagonal constraint $\text{diag}(A) = \text{diag}(B) = \text{diag}(C)$ ensures that the coefficients of each $\ketbra{i}{i} \otimes \ketbra{j}{j}$ basis element are consistent across the triple.

The connection between the triple parameterization and the block structure
described earlier is now transparent: the matrix $A$ encodes the $n \times n$
``diagonal block'' (entries of the form $\ketbra{i}{i} \otimes \ketbra{j}{j}$),
while $B$ and $C$ encode the $n(n-1)/2$ distinct off-diagonal blocks. For each
pair $i < j$, the $4 \times 4$ block on the subspace spanned by $\{\ket{ii}, \ket{ij}, \ket{ji}, \ket{jj}\}$ decomposes as the direct sum
\begin{equation}
  \begin{bmatrix}
    a_{ii} & b_{ij} \\
    b_{ji} & a_{jj}
  \end{bmatrix} \oplus
  \begin{bmatrix}
    a_{ij} & c_{ij} \\
    c_{ji} & a_{ji}
  \end{bmatrix}.
\end{equation}

This structure has important implications for quantum information-theoretic
properties. The positivity of an LDOI density matrix $\rho$
corresponding to $(A,B,C)$ can be characterized through its block-diagonal
structure~\cite{nechita2021diagonal}: (i) the $n \times n$ matrix $B$ restricted
to the diagonal subspace $\{\ket{ii}\}$ must be positive semidefinite, i.e.,
$B \succeq 0$; and (ii) for each $i < j$, the $2 \times 2$ block
$\begin{psmallmatrix} a_{ij} & c_{ij} \\ c_{ji} & a_{ji} \end{psmallmatrix}$
must be positive semidefinite, which requires $a_{ij}, a_{ji} \geq 0$ and
$a_{ij}a_{ji} \geq |c_{ij}|^2$. The condition $a_{ii}a_{jj} \geq |b_{ij}|^2$
for $i \neq j$ follows from $B \succeq 0$. These block-wise constraints are
far simpler to verify than checking positive semidefiniteness of the full
$n^2 \times n^2$ matrix.

Because of this block structure, the only pure states $\rho =
\ketbra{\phi}{\phi}$ that are LDOI are those of the form
\begin{equation}
  \begin{aligned}
    \ket{\phi} & = \sum_{k=1}^n \gamma_k\ket{k} \otimes \ket{k} \quad \text{or} \\
    \ket{\phi} & = \alpha\ket{i} \otimes \ket{j} + \beta\ket{j} \otimes \ket{i} \quad \text{for some} \quad 1 \leq i \neq j \leq n.
  \end{aligned}
\end{equation}

\subsection{Properties of the LDOI subspace}
\label{subsec:LDOI_properties}

The set of LDOI matrices enjoys several important structural properties that
make it amenable to both theoretical analysis and numerical computation. We
summarize the key facts:

\begin{itemize}
  \item \textbf{Vector subspace:} The set of LDOI matrices forms a linear
  subspace of $\Lin(\X \otimes \Y)$. This follows immediately from the
  linearity of the LDOT map: if $\Phi_{\O}(A) = A$ and
  $\Phi_{\O}(B) = B$, then $\Phi_{\O}(\alpha A + \beta B) =
  \alpha A + \beta B$ for any scalars $\alpha, \beta \in \complex$.

  \item \textbf{Orthogonal projection:} The LDOT map $\Phi_{\O}$ is the
  orthogonal projection onto the LDOI subspace with respect to the
  Hilbert-Schmidt inner product $\ip{A}{B} = \tr(A^* B)$. This follows
  from the fact that $\Phi_{\O}$ is self-adjoint
  ($\Phi_{\O}^* = \Phi_{\O}$) and idempotent
  ($\Phi_{\O}^2 = \Phi_{\O}$).

  \item \textbf{Preservation under quantum channels:} If $\rho$ is an LDOI
  state and $\Lambda$ is a quantum channel that commutes with the LDOT map
  (i.e., $\Lambda \circ \Phi_{\O} = \Phi_{\O} \circ \Lambda$),
  then $\Lambda(\rho)$ is also LDOI. Important examples include local
  depolarizing channels and certain entanglement-breaking channels.

  \item \textbf{Dimension reduction:} As established in
  Proposition~\ref{prop:triple_param}, the LDOI subspace has dimension $3n^2 - 2n$ over $\real$ (for Hermitian matrices), compared to dimension $n^4$ for
  the full space $\Herm(\X \otimes \Y)$. This dramatic reduction---from
  quartic to quadratic in $n$---is the key to computational tractability.
\end{itemize}

These properties have important consequences for quantum state distinguishability.
The preservation of LDOI structure under PPT, separable, and LOCC measurements
(as we will show in Section~\ref{sec:LDOI_distinguishability}) means that
when distinguishing LDOI states, we can restrict our search for optimal
measurements to the LDOI subspace. Combined with the dimension reduction, this
leads to semidefinite programs with $O(n^2)$ variables rather than $O(n^4)$,
making exact computation feasible even for moderately large dimensions.

\subsection{Examples of LDOI states}
\label{subsec:LDOI_examples}

We now present several important families of LDOI states, showing their explicit
structure via the triple parameterization. These examples illustrate the
diversity of LDOI states and their relevance to quantum information theory.
Additional examples of LDOI states can also be found in~\cite{nechita2021diagonal}.

\begin{example}[Werner states]\label{exam:werner}
  The Werner states~\cite{werner1989quantum} are a one-parameter family of
  states on $\complex^n \otimes \complex^n$ defined by
  \begin{equation}
    \rho_{\text{W}}(p) = p\frac{\I_{n^2} + F}{n(n+1)} + (1-p)\frac{\I_{n^2} - F}{n(n-1)},
  \end{equation}
  where $F = \sum_{i,j=1}^n \ketbra{ij}{ji}$ is the swap operator and $p \in
  [0,1]$. These states are invariant under all symmetric local unitaries $U
  \otimes U$, which implies they are LDOI. In the triple parameterization, the
  Werner state $\rho_{\text{W}}(p)$ corresponds to
  \begin{equation}
    a_{ij} = \alpha + \beta \delta_{i,j}, \qquad
    b_{ij} = \begin{cases}
      \alpha + \beta & \text{if } i=j,\\
      0 & \text{if } i \neq j,
    \end{cases} \qquad
    c_{ij} = \begin{cases}
      \alpha + \beta & \text{if } i=j,\\
      \beta & \text{if } i \neq j,
    \end{cases}
  \end{equation}
  where $\alpha = \frac{p}{n(n+1)} + \frac{1-p}{n(n-1)}$ and $\beta =
  \frac{p}{n(n+1)} - \frac{1-p}{n(n-1)}$. Equivalently,
  $A = \alpha \1_{n} + \beta \I_n$, $B = (\alpha + \beta)\I_n$, and
  $C = \beta \1_{n} + \alpha \I_n$. Werner states are PPT if and only if $p \geq 1/2$, since the partially
  transposed state has eigenvalues $\alpha$ (multiplicity $n^2-1$) and
  $(2p-1)/n$. Due to their symmetry, Werner states are separable if and only if
  they are PPT, so Werner states with $p \geq 1/2$ are separable. This makes them a convenient test family for our bounds. The
  matrices $A$, $B$, and $C$ are real symmetric, and $\rho_{\text{W}}(p)$ is a
  normalized Hermitian density operator for every real $p \in [0,1]$.
\end{example}

We note that Werner states possess the full $U \otimes U$ symmetry, which is
strictly finer than the local diagonal orthogonal symmetry exploited here. The
$U \otimes U$-invariant subspace is only $2$-dimensional (spanned by the
identity and the swap), whereas the LDOI subspace has dimension $3n^2 - 2n$.
When our LDOI framework is applied to Werner states, the resulting bounds from
the upcoming Theorem~\ref{thm:ldoi_basis_one_copy} are consistent with known optimal
discrimination results, but may not be as tight as those obtainable from the
full $U \otimes U$ reduction. This is the natural trade-off: the LDOI framework
applies uniformly to a much broader class of states, at the cost of not fully
exploiting the additional symmetry present in specific subfamilies.

\begin{example}[X-states]\label{exam:xstates}
  As noted earlier, X-states~\cite{quesada2012quantum} are precisely the LDOI
  states when $n=2$. With entries labeled as
  \begin{equation}
    \rho = \begin{bmatrix}
      \rho_{11} & 0 & 0 & \rho_{14} \\
      0 & \rho_{22} & \rho_{23} & 0 \\
      0 & \rho_{32} & \rho_{33} & 0 \\
      \rho_{41} & 0 & 0 & \rho_{44}
    \end{bmatrix} \in \Density(\complex^4),
  \end{equation}
  the triple parameterization obtained from Equation~\eqref{eq:triple_to_ldoi} is
  \begin{equation}
    A = \begin{bmatrix}
      \rho_{11} & \rho_{22} \\
      \rho_{33} & \rho_{44}
    \end{bmatrix}, \quad
    B = \begin{bmatrix}
      \rho_{11} & \rho_{14} \\
      \rho_{41} & \rho_{44}
    \end{bmatrix}, \quad
    C = \begin{bmatrix}
      \rho_{11} & \rho_{23} \\
      \rho_{32} & \rho_{44}
    \end{bmatrix}.
  \end{equation}
  Many important two-qubit states are X-states, including Bell-diagonal states
  (when $\rho_{11} = \rho_{44}$, $\rho_{22} = \rho_{33}$, $\rho_{14}
  = \rho_{41}$, and $\rho_{23} = \rho_{32}$) and various mixed
  states arising in quantum communication protocols. The entries
  $\rho_{ij}$ may be complex, although density-operator normalization enforces
  $\sum_i \rho_{ii} = 1$ and positivity requires the matrix to be Hermitian.
\end{example}

\begin{example}[Maximally entangled state]\label{exam:maxent}
  The standard maximally entangled state $\ket{\phi^+} = \frac{1}{\sqrt{n}}
  \sum_{i=1}^n \ket{ii}$ has density matrix $\phi^+ = \ketbra{\phi^+}{\phi^+}
  = \frac{1}{n}\sum_{i,j=1}^n \ketbra{i}{j} \otimes \ketbra{i}{j}$,
  which is LDOI. Its triple parameterization is
  \begin{equation}
    A = \frac{1}{n}\I_n, \quad
    B = \frac{1}{n}\1_n, \quad
    C = \frac{1}{n}\I_n,
  \end{equation}
  where $\I_n$ is the $n \times n$ identity matrix and $\1_n$ is the all-ones
  matrix. To verify: the $A$ matrix encodes coefficients of
  $\ketbra{i}{i} \otimes \ketbra{j}{j}$, which only appear when $i=j$ in $\phi^+$
  (giving $a_{ii} = 1/n$, $a_{ij} = 0$ for $i \neq j$). The $B$ matrix encodes
  coefficients of $\ketbra{i}{j} \otimes \ketbra{i}{j}$ for $i \neq j$, all of
  which equal $1/n$. The $C$ matrix encodes coefficients of
  $\ketbra{i}{j} \otimes \ketbra{j}{i}$ for $i \neq j$, none of which appear
  in $\phi^+$ (giving $c_{ij} = 0$ for $i \neq j$).
  This state is maximally entangled and hence not separable, but it can be
  perfectly distinguished from any orthogonal state by a global measurement.
\end{example}

\begin{example}[Product states]\label{exam:product}
  Any product state $\rho = \rho_A \otimes \rho_B$ where both $\rho_A$ and
  $\rho_B$ are diagonal in the computational basis is LDOI. For instance, the
  product basis states $\ket{ij}$ for $1 \leq i,j \leq n$ are LDOI. The state
  $\ketbra{ij}{ij}$ has triple parameterization
  \begin{equation}
    A_{k\ell} = \begin{cases}
      1 & \text{if } k=i, \ell=j, \\
      0 & \text{otherwise},
    \end{cases}\qquad
    B_{k\ell} = C_{k\ell} = \begin{cases}
      a_{kk} & \text{if } k = \ell, \\
      0 & \text{otherwise},
    \end{cases}
  \end{equation}
  so that $B$ and $C$ are diagonal matrices satisfying the constraint
  $\text{diag}(A) = \text{diag}(B) = \text{diag}(C)$. Product bases are
  perfectly distinguishable by local measurements, demonstrating that not all
  LDOI bases exhibit nonlocal distinguishability phenomena.
  The matrices appearing in this example are real, and the basis vectors
  $\ket{ij}$ are normalized product states.
\end{example}

These examples demonstrate the breadth of LDOI states: from maximally entangled
(Example~\ref{exam:maxent}) to separable (Example~\ref{exam:product}), and from
well-studied families with known separability criteria
(Example~\ref{exam:werner}) to general parametric classes
(Example~\ref{exam:xstates}). In the context of state distinguishability, LDOI
states are particularly amenable to analysis due to their reduced parameter
space and tractable SDP characterizations.

\subsection{Orthonormal LDOI bases}
\label{subsec:LDOI_bases}

We now characterize which orthonormal bases of $\X \otimes \Y$ consist entirely
of LDOI pure states. As the preceding discussion shows, such bases must contain
exactly $n$ states supported on the diagonal subspace (of the form $\sum_k \gamma_k \ket{kk}$)
and $n(n-1)$ Schmidt-rank-2 states (of the form $\alpha \ket{ij} + \beta
\ket{ji}$ with $1 \leq i \neq j \leq n$). The diagonal states are maximally entangled only when all $|\gamma_k| = 1/\sqrt{n}$;
in general they have Schmidt rank at most $n$. The following definition makes this structure precise.

\begin{definition}\label{defn:LDOI_basis}
  Let $n \geq 2$ be an integer, let $\X = \Y = \complex^n$ be complex Euclidean
  spaces. We say that a set $\eta \subset \X \otimes \Y$ is an
  \emph{orthonormal LDOI basis} of $\X \otimes \Y$ if there exist $U \in
  \Unitary(\X)$ and $A \in \Lin(\X)$ with $|a_{i,j}|^2 + |a_{j,i}|^2 = 1$ for
  all $1 \leq i \neq j \leq n$, such that $\eta = \{\ket{\phi_{i,j}} : 1 \leq
  i,j \leq n\}$, where
  \begin{equation}
    \begin{aligned}
      \ket{\phi_{i,i}} & = \sum_{k=1}^n u_{i,k}\ket{k} \otimes \ket{k} \quad \text{for all} \quad 1 \leq i \leq n, \\
      \ket{\phi_{i,j}} & = a_{i,j}\ket{i} \otimes \ket{j} + \overline{a_{j,i}}\ket{j} \otimes \ket{i} \quad \text{for all} \quad 1 \leq i < j \leq n, \quad \text{and} \\
      \ket{\phi_{j,i}} & = a_{j,i}\ket{i} \otimes \ket{j} - \overline{a_{i,j}}\ket{j} \otimes \ket{i} \quad \text{for all} \quad 1 \leq i < j \leq n.
    \end{aligned}
  \end{equation}
\end{definition}

Indeed, in the above definition each $\ket{\phi_{k,k}}$ is immediately
orthogonal to each $\ket{\phi_{i,j}}$ when $i \neq j$ since they have disjoint
supports, unitarity of $U$ is equivalent to the fact that the
$\ket{\phi_{k,k}}$'s have unit length and are orthogonal to each other, and the
constraint $|a_{i,j}|^2 + |a_{j,i}|^2 = 1$ is equivalent to the fact that each
$\ket{\phi_{i,j}}$ has unit length when $i \neq j$. Note that the diagonal
entries of the matrix $A$ are unconstrained and do not appear in the
parameterization of $\eta$, as they play no role in defining the basis states.
An orthonormal LDOI basis of $\X \otimes \Y$ contains exactly $n^2$ vectors, so
in the terminology of the preliminaries we have $N = n^2$ measurement operators
when distinguishing such bases.

\begin{remark}[Orthonormal CLDUI bases]
Numerous variants of Definition~\ref{defn:LDOI_basis} are possible, since numerous variants of LDOI states are possible. For example, if we instead considered matrices that are \emph{conjugate locally diagonal unitarily invariant (CLDUI)} throughout this work (as in \cite{johnston2019pairwise}), then we could analogously define \emph{orthonormal CLDUI bases} of $\X \otimes \Y$ to be just as in Definition~\ref{defn:LDOI_basis}, but with $A$ upper triangular (so that each $\ket{\phi_{i,j}}$ is a pure product state).

Our reason for working with LDOI states (and the LDOI twirl, and orthonormal LDOI bases) in this work is that it is the largest family of states we are aware of that arises in a natural way from a local twirl, so we get more general results this way. All of our results concerning distinguishability of orthonormal LDOI bases apply automatically to orthonormal CLDUI bases, for example, since CLDUI states are LDOI.
\end{remark}

\begin{example}\label{exam:bell_basis}
  When $\X = \Y = \complex^2$, the Bell basis consists of the states
  \begin{equation}
    \begin{aligned}
      \ket{\phi_{1,1}} & = \frac{1}{\sqrt{2}}(\ket{11} + \ket{22}), & 
      \ket{\phi_{1,2}} & = \frac{1}{\sqrt{2}}(\ket{12} + \ket{21}), \\
      \ket{\phi_{2,1}} & = \frac{1}{\sqrt{2}}(\ket{12} - \ket{21}), & 
      \ket{\phi_{2,2}} & = \frac{1}{\sqrt{2}}(\ket{11} - \ket{22}).
    \end{aligned}
  \end{equation}
  The Bell basis is an orthonormal LDOI basis in the sense of
  Definition~\ref{defn:LDOI_basis} with
  \begin{equation}
    A = U = \frac{1}{\sqrt{2}}
    \begin{bmatrix}
      1 & 1 \\ 1 & -1
    \end{bmatrix}.
  \end{equation}
  At the other extreme, the product basis
  $\{\ket{11},\ket{12},\ket{21},\ket{22}\}$ is also an orthonormal LDOI basis,
  with
  \begin{equation}
    A = \begin{bmatrix}
      0 & 0 \\ 1 & 0
    \end{bmatrix} \quad \text{and} \quad U = \begin{bmatrix}
      1 & 0 \\ 0 & 1
    \end{bmatrix}.
  \end{equation}
  In both of these examples, $A$ is non-unique since its diagonal entries are
  irrelevant.
\end{example}

\section{Distinguishability of LDOI states}
\label{sec:LDOI_distinguishability}

We begin by establishing a simple but important result: when distinguishing
LDOI states, we may restrict to LDOI measurement operators without loss of
generality.

\begin{theorem}\label{thm:ldoi_sdp_simplifies}
  Suppose $\{\rho_1,\ldots,\rho_N\}$ is a set of LDOI quantum states. Then the
  optimal value of the semidefinite program~\eqref{eq:ppt-primal} for PPT
  measurements and the analogous convex optimization problem for separable
  measurements are attained by LDOI measurement operators. Furthermore, the
  supremum over LOCC measurements is unchanged and can be approached
  arbitrarily closely by LDOI LOCC POVMs.
\end{theorem}

\begin{proof}
  For PPT and separable measurements, the result follows from the
  self-adjointness of the LDOT map $\Phi_{\O}$. If each $\rho_k$ is LDOI then
  \begin{equation}\label{eq:ldoi_self_dual}
    \ip{P_k}{\rho_k} = \ip{P_k}{\Phi_{\O}(\rho_k)} = \ip{\Phi_{\O}(P_k)}{\rho_k},
  \end{equation}
  where the second equality uses $\Phi_{\O}^* = \Phi_{\O}$.
  Furthermore, the operators $\Phi_{\O}(P_1)$, $\ldots$,
  $\Phi_{\O}(P_N)$ satisfy all of the same constraints that $P_1$, $\ldots$, $P_N$ do (since
  $P_k \in \ppt(\X : \Y)$ implies $\Phi_{\O}(P_k) \in \ppt(\X : \Y)$, and
  similarly for separability). In particular, $\Phi_{\O}$ preserves the
  completeness relation for POVMs: since $\Phi_{\O}$ is linear and unital
  (each $O \in \DO(\X)$ satisfies $(O \otimes O)^2 = \I$), we have
  \begin{equation}\label{eq:ldoi_completeness}
    \sum_{k=1}^N \Phi_{\O}(P_k) = \Phi_{\O}\left(\sum_{k=1}^N P_k\right)
    = \Phi_{\O}(\I_{\X \otimes \Y}) = \I_{\X \otimes \Y},
  \end{equation}
  so $\{\Phi_{\O}(P_k)\}$ is again a POVM.

  For LOCC measurements, we use a shared-randomness argument. Given any LOCC
  POVM $\{P_k\}$, consider the following protocol: using shared randomness,
  Alice and Bob sample a diagonal sign matrix $O$ uniformly at random,
  apply the local unitaries $O \otimes O$, then execute the original LOCC
  protocol for $\{P_k\}$. The effective POVM elements are
  $\Phi_{\O}(P_k) = \frac{1}{2^n}\sum_O (O \otimes O) P_k (O \otimes O)$,
  which remain LOCC since LOCC is closed under local unitaries and shared
  randomness. For LDOI states $\rho_k$, the success probability is unchanged
  by~\eqref{eq:ldoi_self_dual}, so the LOCC supremum can be approached by
  LDOI LOCC POVMs.
\end{proof}

Theorem~\ref{thm:ldoi_sdp_simplifies} shows that when distinguishing LDOI
states, we can restrict our search for optimal measurements to LDOI
measurements without loss of generality. This has significant computational
implications: for an ensemble of LDOI states on $\complex^n \otimes \complex^n$,
the standard PPT semidefinite program requires optimizing over $n^4$ variables,
but by exploiting the block structure of LDOI matrices
(Subsection~\ref{subsec:LDOI_structure}), the optimization can be restricted
to the $3n^2-2n$ dimensional LDOI subspace. For even moderately large
dimensions (e.g., $n = 10$), this represents a reduction from $10{,}000$
variables to $280$ variables---more than a $30$-fold decrease in problem size.
This dimensional reduction makes previously intractable problems computationally
feasible.

Since LDOI states are invariant under local diagonal orthogonal symmetries, any
optimization (like computing the optimal success probabilities) can be
restricted to the LDOI subspace, drastically reducing the problem size. Our
main result of this section establishes bounds on the probability of
successfully distinguishing LDOI bases via LOCC, separable, and PPT
measurements.

The following theorem applies the $(U,A)$ parameterization of
Definition~\ref{defn:LDOI_basis} to uniform ensembles over orthonormal LDOI
bases, yielding computable lower and upper bounds on the optimal success
probabilities.

\begin{theorem}\label{thm:ldoi_basis_one_copy}
  Let $n \geq 2$ be an integer, let $\X = \Y = \complex^n$, and let $A \in \Lin(\X)$
  and $U \in \Unitary(\X)$ parameterize an orthonormal LDOI basis of $\X \otimes \Y$
  as in Definition~\ref{defn:LDOI_basis}. Define the ensemble $\E = \{(1/n^2,
  \ketbra{\phi_{i,j}}{\phi_{i,j}}) : 1 \leq i,j \leq n\}$ with uniform prior
  distribution, where $\ket{\phi_{i,j}}$ are the basis vectors from
  Definition~\ref{defn:LDOI_basis}. The success probability of correctly distinguishing the states in
  $\E$ via LOCC satisfies
  \begin{align}\label{eq:ldoi_lb}
    \opt_{\locc}(\E) & \geq \frac{1}{n^2}\left(\sum_{\substack{i,j=1 \\ i \neq j}}^n \max\left\{ |a_{i,j}|^2, |a_{j,i}|^2 \right\} + \max_{\sigma \in S_n}\left\{ \sum_{i=1}^n |u_{i,\sigma(i)}|^2\right\}\right).
  \end{align}
  Furthermore, if $c_1,\ldots,c_n$ are any real numbers for which
  \begin{align}\label{eq:ldoi_c_bound}
    c_i \geq \max_{1 \leq k \leq n} \left\{ |u_{k,i}|^2 \right\} \quad \text{and} \quad c_ic_j \geq |a_{i,j}|^2|a_{j,i}|^2
  \end{align}
  for all $1 \leq i \leq n$ (first inequality) and all $1 \leq i < j \leq n$ (second inequality), then
  \begin{align}\label{eq:ldoi_ub}
    \opt_{\ppt}(\E) & \leq \frac{1}{n^2}\left(\sum_{\substack{i,j=1 \\ i \neq j}}^n \max\left\{ |a_{i,j}|^2, |a_{j,i}|^2 \right\} + \sum_{i=1}^n c_i\right).
  \end{align}
\end{theorem}

Before we prove the above theorem, we make numerous remarks about the bounds
that it provides:

\begin{itemize}
  \item The lower bound~\eqref{eq:ldoi_lb} and the upper bound~\eqref{eq:ldoi_ub} are very close to each other (see Corollary~\ref{cor:close_LDOI_LOCC_PPT}), and often even equal to each other (see Corollary~\ref{cor:large_U}). In particular, the initial term $\sum_{\substack{i,j=1 \\ i \neq j}}^n \max\left\{ |a_{i,j}|^2, |a_{j,i}|^2 \right\}$ is the same in both bounds, and it is just the second summation in each bound that might differ.
  
  \item The term $\max_{\sigma \in S_n}\left\{ \sum_{i=1}^n |u_{i,\sigma(i)}|^2\right\}$ in the bound~\eqref{eq:ldoi_lb} perhaps looks difficult to compute. However, it is actually an instance of the Assignment Problem, which can be solved in polynomial time by the Hungarian algorithm~\cite{kuhn1955hungarian}. The original formulation of that algorithm runs in $O(n^4)$ time~\cite{munkres1957algorithms}, and shortest-augmenting-path implementations that run in $O(n^3)$ time were subsequently obtained~\cite{tomizawa1971some, edmonds1972theoretical}.
  
  \item Similarly, the optimal $c_i$'s for the bound~\eqref{eq:ldoi_ub} can be found via semidefinite programming, since the constraint $c_ic_j \geq |a_{i,j}|^2|a_{j,i}|^2$ is equivalent to the matrix
    \begin{equation}
    \begin{bmatrix}
      c_i & a_{i,j}\overline{a_{j,i}} \\
      \overline{a_{i,j}}a_{j,i} & c_j
    \end{bmatrix}
  \end{equation}
  being positive semidefinite. This gives a much smaller (and thus much faster
  to evaluate numerically) semidefinite program than the general semidefinite
  program~\eqref{eq:ppt-dual} that is used to find upper bounds
  on $\opt_{\ppt}(\E)$. However, it is not always the case that there
  exist $c_i$'s for which equality is attained in bound~\eqref{eq:ldoi_ub}
  (see Example~\ref{exam:PPT_bound_not_attained}).
  
  \item One feasible choice for the $c_i$'s is $c_i = \max\left\{ \frac{1}{2}, \max_{1 \leq k \leq n} \left\{ |u_{k,i}|^2 \right\} \right\}$, since $c_i,c_j \geq 1/2$ ensures $c_ic_j \geq 1/4 \geq |a_{i,j}|^2|a_{j,i}|^2$, with the final inequality following from the fact that $|a_{i,j}|^2 + |a_{j,i}|^2 = 1$. This gives the following explicit (but slightly weaker than~\eqref{eq:ldoi_ub}) upper bound for $\opt_{\ppt}(\E)$
  \begin{align}\label{eq:ppt_dist_weaker}
    \opt_{\ppt}(\E) & \leq \frac{1}{n^2}\left(\sum_{\substack{i,j=1 \\ i \neq j}}^n \max\left\{ |a_{i,j}|^2, |a_{j,i}|^2 \right\} + \sum_{i=1}^n \max\left\{ \frac{1}{2}, \max_{1 \leq k \leq n} \left\{ |u_{k,i}|^2 \right\} \right\}\right).
  \end{align}
  
  \item The bound~\eqref{eq:ppt_dist_weaker} shows that an orthonormal LDOI basis is perfectly distinguishable by PPT measurements if and only if it consists entirely of product states. To see this, notice that $\opt_{\ppt}(\E) = 1$ implies $\max\left\{ |a_{i,j}|^2, |a_{j,i}|^2 \right\} = 1$ and $\max_{1 \leq k \leq n} \left\{ |u_{k,i}|^2 \right\} = 1$ for all $1 \leq i \neq j \leq n$. This implies that every entry in the matrices $A$ and $U$ from Definition~\ref{defn:LDOI_basis} has magnitude equal to $0$ or $1$, giving the result.
  
  \item The bound~\eqref{eq:ldoi_lb} immediately implies a universal lower bound on LOCC distinguishability for all LDOI bases (see Corollary~\ref{cor:universal_bound} below).

\item When $a_{i,j} = 1/\sqrt{2}$ for all $i \neq j$, numerical evidence (via the software described in Section~\ref{sec:software}) strongly suggests that equality holds in the bound~\eqref{eq:ldoi_ub}. That is, there exist $c_1,
  \ldots, c_n$ satisfying the constraints~\eqref{eq:ldoi_c_bound}
  such that $\opt_{\ppt}(\E)$ exactly equals the
  upper bound. This has been verified numerically for various choices of $U$
  (including identity, Hadamard, and random unitary matrices) in
  dimensions $n = 2, 3, 4$), and for the Fourier matrix in all dimensions (see Example~\ref{exam:mixed_bell_basis}).
\end{itemize}

We note that Theorem~\ref{thm:ldoi_basis_one_copy} and its corollaries are
stated for the uniform prior distribution $p_i = 1/n^2$. For non-uniform
priors, the structural reduction of
Theorem~\ref{thm:ldoi_sdp_simplifies} still applies: optimal PPT and separable
measurements can be taken to be LDOI regardless of the prior. However, the
explicit bounds change character. The LOCC lower bound~\eqref{eq:ldoi_lb}
generalizes by weighting each term by its prior probability, but the resulting
assignment problem is no longer symmetric and may not admit a closed-form
solution. The PPT upper bound~\eqref{eq:ldoi_ub} similarly generalizes with
prior-weighted diagonal entries, though the constraints on the $c_i$ become
prior-dependent. Extending these bounds to obtain tight results for non-uniform
priors remains a direction for future work.

\begin{proof}[Proof of Theorem~\ref{thm:ldoi_basis_one_copy}]
To prove the lower bound on $\opt_{\locc}(\E)$, let $\sigma \in S_n$ and consider the local measurement operators
  \begin{align}\begin{split}\label{eq:local_P_ldoi}
  P_{i,i} & = \ketbra{\sigma(i)}{\sigma(i)} \otimes \ketbra{\sigma(i)}{\sigma(i)} \quad \text{for} \quad 1 \leq i \leq n, \\
    P_{i,j} & = \begin{cases}
      \ketbra{i}{i} \otimes \ketbra{j}{j} \quad \text{if } |a_{i,j}| \geq |a_{j,i}| \\
      \ketbra{j}{j} \otimes \ketbra{i}{i} \quad \text{otherwise}
    \end{cases} \quad \text{for} \quad 1 \leq i < j \leq n, \quad \text{and} \\
    P_{j,i} & = \begin{cases}
      \ketbra{j}{j} \otimes \ketbra{i}{i} \quad \text{if } |a_{i,j}| \geq |a_{j,i}| \\
      \ketbra{i}{i} \otimes \ketbra{j}{j} \quad \text{otherwise}
    \end{cases} \quad \text{for} \quad 1 \leq i < j \leq n.
  \end{split}\end{align}
  
  A direct calculation shows that $\sum_{i,j = 1}^n P_{i,j} = \I$,
  \begin{equation}
    \ip{P_{i,i}}{\phi_{i,i}} = |u_{i,\sigma(i)}|^2 \quad \text{and} \quad \ip{P_{i,j}}{\phi_{i,j}} = \max\left\{ |a_{i,j}|^2, |a_{j,i}|^2 \right\},
  \end{equation}
  for all $1 \leq i,j \leq n$. The lower bound~\eqref{eq:ldoi_lb} now follows by maximizing over all $\sigma \in S_n$.
  
  To prove the upper bound on $\opt_{\ppt}(\E)$, for each pair $(i,j)$ with
  $1 \leq i,j \leq n$, let $\gamma_{i,j}$ be a non-negative real number. Define the
  Hermitian operator
  \begin{equation}
    H = \frac{1}{n^2}\left(\sum_{i, j = 1}^n \gamma_{i,j}Q_{i,j}\right) \in \Herm(\X \otimes \Y),
  \end{equation}
  where $Q_{i,j} = \ketbra{i}{i} \otimes \ketbra{j}{j}$. It is clear that $H$ is diagonal with non-negative entries, and is thus positive semidefinite. We now investigate which values of $\{\gamma_{i,j}\}$ result in $H$ being a feasible point of the semidefinite program~\eqref{eq:ppt-dual-constrained}, i.e.,
  \begin{equation}
    H - \frac{1}{n^2}\pt_{\X}\left(\phi_{i,j}\right) \in \Pos(\X \otimes \Y)
  \end{equation}
  for all $1 \leq i,j \leq n$.
  
  A calculation reveals that
  \begin{align}\begin{split}\label{eq:ldoi_ppt_ub_1x2_2x2}
    H - \frac{1}{n^2} \pt_{\X}\left(\phi_{k,k}\right) & =
    \frac{1}{n^2}\left(\sum_{i, j = 1}^n \gamma_{i,j} Q_{i,j} - \sum_{i,j=1}^n u_{k,i}\overline{u_{k,j}}\ketbra{j}{i} \otimes \ketbra{i}{j}\right) \\
     & = \frac{1}{n^2}\left(\sum_{i=1}^n (\gamma_{i,i} - |u_{k,i}|^2)Q_{i,i} + \sum_{\substack{i,j=1 \\ i \neq j}}^n \Big( \gamma_{i,j} Q_{i,j} - u_{k,i}\overline{u_{k,j}}\ketbra{j}{i} \otimes \ketbra{i}{j}\Big) \right).
  \end{split}\end{align}
  The bottom line of Equation~\eqref{eq:ldoi_ppt_ub_1x2_2x2} is a decomposition of $H - \pt_{\X}\left(\phi_{k,k}\right)/n^2$ as a direct sum of $1 \times 1$ and $2 \times 2$ blocks, so it is positive semidefinite if and only if $\gamma_{i,i} \geq |u_{k,i}|^2$ and $\gamma_{i,j}\gamma_{j,i} \geq |u_{k,i}|^2|u_{k,j}|^2$ for all $1 \leq i,j \leq n$, where the latter condition follows from the Schur complement criterion for the $2 \times 2$ blocks.
  
  Similarly, if $1 \leq k < \ell \leq n$ then
  \begin{align}\begin{split}\label{eq:ldoi_ppt_ub_1x2_2x2b}
    H - \frac{1}{n^2} \pt_{\X}\left(\phi_{k,\ell}\right) & =
    \frac{1}{n^2}\left(\sum_{i, j = 1}^n \gamma_{i,j} Q_{i,j} - |a_{k,\ell}|^2Q_{k,\ell} - |a_{\ell,k}|^2Q_{\ell,k} - a_{k,\ell}a_{\ell,k}\ketbra{\ell\ell}{kk} - \overline{a_{k,\ell}a_{\ell,k}}\ketbra{kk}{\ell\ell}\right) \\
    & = \frac{1}{n^2}\left(\sum_{\substack{i, j = 1\\ i,j \notin \{k,\ell\}}}^n \gamma_{i,j} Q_{i,j}\right) + \frac{1}{n^2}\Big((\gamma_{k,\ell} - |a_{k,\ell}|^2)Q_{k,\ell} + (\gamma_{\ell,k} - |a_{\ell,k}|^2)Q_{\ell,k} \\
    & \qquad \qquad \qquad \qquad + \gamma_{k,k}Q_{k,k} + \gamma_{\ell,\ell}Q_{\ell,\ell} - a_{k,\ell}a_{\ell,k}\ketbra{\ell\ell}{kk} - \overline{a_{k,\ell}a_{\ell,k}}\ketbra{kk}{\ell\ell}\Big).
  \end{split}\end{align}
  The bottom line of Equation~\eqref{eq:ldoi_ppt_ub_1x2_2x2b} is a decomposition of $H - \pt_{\X}\left(\phi_{k,\ell}\right)/n^2$ as a direct sum of $1 \times 1$ and $2 \times 2$ blocks, so it is positive semidefinite if and only if $\gamma_{k,\ell} \geq |a_{k,\ell}|^2$, $\gamma_{\ell,k} \geq |a_{\ell,k}|^2$, and $\gamma_{k,k}\gamma_{\ell,\ell} \geq |a_{k,\ell}|^2|a_{\ell,k}|^2$. A similar computation of $H - \pt_{\X}\left(\phi_{\ell,k}\right)/n^2$ when $1 \leq k < \ell \leq n$ shows that we also need $\gamma_{k,\ell} \geq |a_{\ell,k}|^2$.
  
  Altogether, this means that $H$ is a feasible point of the semidefinite program~\eqref{eq:ppt-dual-constrained} if and only if the scalars $\{\gamma_{i,j}\}$ satisfy the following constraints for all $1 \leq i,j \leq n$:
  \begin{equation}
    \begin{aligned}
      \gamma_{i,i} & \geq \max_{1 \leq k \leq n} \left\{ |u_{k,i}|^2 \right\}, & \gamma_{i,i}\gamma_{j,j} & \geq |a_{i,j}|^2|a_{j,i}|^2, \\
      \gamma_{i,j} & \geq \max\left\{ |a_{i,j}|^2, |a_{j,i}|^2 \right\}, & \gamma_{i,j}\gamma_{j,i} & \geq \max_{1 \leq k \leq n} \left\{ |u_{k,i}|^2|u_{k,j}|^2 \right\}.
    \end{aligned}
  \end{equation}
  Notice that the bottom-right constraint is redundant. To see this, observe
  that $\gamma_{i,j} \geq \max\left\{ |a_{i,j}|^2, |a_{j,i}|^2 \right\} \geq
  1/2$ (using $|a_{i,j}|^2 + |a_{j,i}|^2 = 1$), and similarly $\gamma_{j,i}
  \geq 1/2$. Therefore $\gamma_{i,j}\gamma_{j,i} \geq 1/4$. Since $U$ is
  unitary, we have $|u_{k,i}|^2|u_{k,j}|^2 \leq 1/4$ for all $i,j,k$ (when $i
  \neq j$), which implies the bottom-right constraint is automatically
  satisfied. In particular, this means that we can choose
  \begin{equation}
    \gamma_{i,j} = \max\left\{ |a_{i,j}|^2, |a_{j,i}|^2 \right\}
  \end{equation}
  for all $1 \leq i \neq j \leq n$. Setting $c_i = \gamma_{i,i}$ then shows
  that the objective value in the semidefinite
  program~\eqref{eq:ppt-dual-constrained} equals
  \begin{equation}
    \tr(H) = \frac{1}{n^2}\sum_{i,j=1}^n \gamma_{i,j} = \frac{1}{n^2}\left(\sum_{\substack{i,j=1 \\ i \neq j}}^n \max\left\{ |a_{i,j}|^2, |a_{j,i}|^2 \right\} + \sum_{i=1}^n c_i\right)
  \end{equation}
  which is exactly the desired upper bound~\eqref{eq:ldoi_ub}.
\end{proof}

While the bounds of Theorem~\ref{thm:ldoi_basis_one_copy}
perhaps look a bit complicated, there is a wide class of LDOI
bases for which the bounds simplify considerably, and we actually get equality:
whenever each column of $U$ contains a sufficiently large entry, the two
expressions collapse to the same value.

\begin{cor}\label{cor:large_U}
  Let $n \geq 2$ be an integer, let $\X = \Y = \complex^n$, and let $A \in \Lin(\X)$
  and $U \in \Unitary(\X)$ parameterize an orthonormal LDOI basis of $\X \otimes \Y$
  as in Definition~\ref{defn:LDOI_basis}. Define the ensemble $\E = \{(1/n^2,
  \ketbra{\phi_{i,j}}{\phi_{i,j}}) : 1 \leq i,j \leq n\}$ with uniform prior
  distribution, where $\ket{\phi_{i,j}}$ are the basis vectors from
  Definition~\ref{defn:LDOI_basis}. Suppose further that there exists a permutation
  $\sigma \in S_n$ such that $|u_{\sigma(i),i}| \geq 1/\sqrt{2}$ for every $1 \leq i \leq n$ (i.e., each row and column of $U$ has an entry with magnitude at least $1/\sqrt{2}$). Then
  \begin{equation}
    \opt_{\locc}(\E) = \opt_{\sep}(\E) = \opt_{\ppt}(\E) = \frac{1}{n^2}\left(\sum_{\substack{i,j=1 \\ i \neq j}}^n \max\left\{ |a_{i,j}|^2, |a_{j,i}|^2 \right\} + \sum_{i=1}^n \left( \max_{1 \leq k \leq n} \left\{ |u_{k,i}|^2 \right\}\right)\right).
  \end{equation}
\end{cor}

\begin{proof}
  By assumption there is $\sigma \in S_n$ with $|u_{\sigma(i),i}|^2 \geq 1/2$ for
  all $i$. For each column define
  \begin{equation}
    c_i = \max_{1 \leq k \leq n} \left\{ |u_{k,i}|^2 \right\}
  \end{equation}
  for $1 \leq i \leq n$, then we have $c_i \geq 1/2$. It follows that $c_ic_j
  \geq 1/4 \geq |a_{i,j}|^2|a_{j,i}|^2$, with the latter inequality following
  from the fact that $|a_{i,j}|^2 + |a_{j,i}|^2 = 1$. In other words, the
  second inequality of~\eqref{eq:ldoi_c_bound} follows from the
  first.
  
  Under this choice of $c_i$, the upper bound~\eqref{eq:ldoi_ub}
  becomes
  \begin{equation}\label{eq:ldoi_better_ppt_bound}
    \opt_{\ppt}(\E) \leq \frac{1}{n^2}\left(\sum_{\substack{i,j=1 \\ i \neq j}}^n \max\left\{ |a_{i,j}|^2, |a_{j,i}|^2 \right\} + \sum_{i=1}^n \left( \max_{1 \leq k \leq n} \left\{ |u_{k,i}|^2 \right\}\right)\right).
  \end{equation}

  For the lower bound~\eqref{eq:ldoi_lb}, we evaluate the permutation $\sigma$
  guaranteed by the hypothesis to obtain
  \begin{equation}
    \max_{\sigma' \in S_n}\left\{ \sum_{i=1}^n |u_{i,\sigma'(i)}|^2\right\}
    \geq \sum_{i=1}^n |u_{\sigma(i),i}|^2
    = \sum_{i=1}^n \left(\max_{1 \leq k \leq n} \left\{ |u_{k,i}|^2 \right\}\right),
  \end{equation}
  since each $|u_{\sigma(i),i}|^2$ is one of the maximizers for column $i$. Thus
  the lower bound specializes to
  \begin{equation} \label{eq:ldoi_better_locc_bound}
    \opt_{\locc}(\E) \geq \frac{1}{n^2}\left(\sum_{\substack{i,j=1 \\ i \neq j}}^n \max\left\{ |a_{i,j}|^2, |a_{j,i}|^2 \right\} + \sum_{i=1}^n \left( \max_{1 \leq k \leq n} \left\{ |u_{k,i}|^2 \right\}\right)\right).
  \end{equation}
  Since the bounds~\eqref{eq:ldoi_better_ppt_bound}
  and~\eqref{eq:ldoi_better_locc_bound} are the
  same as each other, the result follows.
\end{proof}

In the two-qubit case, all LDOI bases satisfy the hypotheses of
Corollary~\ref{cor:large_U}, so we get the following simple
result:

\begin{cor}\label{cor:ppt_dist_2x2}
  Let $\X = \Y = \complex^2$ and let $\alpha,\beta,\gamma,\delta \in \complex$ be such
  that $|\alpha|^2 + |\beta|^2 = |\gamma|^2 + |\delta|^2 = 1$. Define the ensemble
  $\E = \{(1/4, \ketbra{\phi_i}{\phi_i}) : 1 \leq i \leq 4\}$ with uniform prior
  distribution, where
  \begin{equation}
    \begin{aligned}
      \ket{\phi_1} & = \alpha\ket{11} + \overline{\beta}\ket{22}, & \ket{\phi_2} & = \beta\ket{11} - \overline{\alpha}\ket{22}, \\
      \ket{\phi_3} & = \gamma\ket{12} + \overline{\delta}\ket{21}, & \ket{\phi_4} & = \delta\ket{12} - \overline{\gamma}\ket{21}.
    \end{aligned}
  \end{equation}
  Then
  \begin{equation}
    \opt_{\locc}(\E) = \opt_{\sep}(\E) = \opt_{\ppt}(\E) = \frac{1}{2}\Big(\max\left\{ |\alpha|^2, |\beta|^2 \right\} + \max\left\{ |\gamma|^2, |\delta|^2 \right\}\Big).
  \end{equation}
\end{cor}

\begin{proof}
  The states $\{\ket{\phi_i}\}$ form an orthonormal LDOI basis, and every
  $2 \times 2$ unitary matrix has the property that each column has an entry
  with magnitude at least $1/\sqrt{2}$. Therefore Corollary~\ref{cor:large_U}
  applies and gives the claimed formula immediately.
\end{proof}

\begin{remark}
  When $|\alpha| = |\beta| = |\gamma| = |\delta| = 1/\sqrt{2}$, the ensemble $\E$ in Corollary~\ref{cor:ppt_dist_2x2} reduces to the uniform ensemble over the standard Bell basis, for which $\opt_{\locc}(\E) = \opt_{\sep}(\E) = \opt_{\ppt}(\E) = 1/2$.
  More generally, the parametrized families of states in Corollary~\ref{cor:ppt_dist_2x2} generalize the parametrized Bell states studied in Example~3 of~\cite{bandyopadhyay2013tight} and in~\cite{bandyopadhyay2015limitations}.
  Corollary~\ref{cor:ppt_dist_2x2} shows that the explicit formula for optimal distinguishability holds not just for these specific parametrized families, but for \emph{all} orthonormal LDOI bases in the two-qubit setting.
\end{remark}

\begin{cor}\label{cor:bell_optimal}
  Among all two-qubit orthonormal LDOI bases, the uniform ensemble over the Bell
  basis achieves the minimum LOCC (and PPT and separable) distinguishability.
  That is, for any two-qubit LDOI basis with uniform ensemble $\E = \{(1/4,
  \ketbra{\phi_i}{\phi_i}) : 1 \leq i \leq 4\}$,
  \begin{equation}
    \opt_{\locc}(\E) \geq \frac{1}{2},
  \end{equation}
  with equality if and only if the basis is the Bell basis (up to local unitaries).
\end{cor}

\begin{proof}
  By Corollary~\ref{cor:ppt_dist_2x2}, for any two-qubit LDOI basis with uniform ensemble $\E$ we have
  \begin{equation}
    \opt_{\locc}(\E) = \frac{1}{2}\Big(\max\left\{ |\alpha|^2, |\beta|^2 \right\} + \max\left\{ |\gamma|^2, |\delta|^2 \right\}\Big),
  \end{equation}
  where $|\alpha|^2 + |\beta|^2 = |\gamma|^2 + |\delta|^2 = 1$. Since $\max\{a,
  b\} \geq 1/2$ when $a+b=1$ and $a,b \geq 0$, with equality if and only if $a
  = b = 1/2$, we have
  \begin{equation}
    \max\left\{ |\alpha|^2, |\beta|^2 \right\} \geq \frac{1}{2} \quad \text{and} \quad \max\left\{ |\gamma|^2, |\delta|^2 \right\} \geq \frac{1}{2},
  \end{equation}
  with equality in both inequalities if and only if $|\alpha| = |\beta| =
  |\gamma| = |\delta| = 1/\sqrt{2}$, which corresponds to the Bell basis.
\end{proof}

\begin{cor}\label{cor:universal_bound}
  For any orthonormal LDOI basis of $\X \otimes \Y$ where $\X = \Y = \complex^n$, let
  $\E = \{(1/n^2, \ketbra{\phi_{i,j}}{\phi_{i,j}}) : 1 \leq i,j \leq n\}$ denote the
  uniform ensemble, where $\ket{\phi_{i,j}}$ are the basis vectors. Then
  \begin{equation}
    \opt_{\locc}(\E) \geq \frac{1}{2} - \frac{n - 2}{2n^2}.
  \end{equation}
  In particular, this function is minimized when $n = 4$, giving the dimension-independent bound
  \begin{equation}
    \opt_{\locc}(\E) \geq \frac{7}{16} = 0.4375
  \end{equation}
  for all uniform ensembles over orthonormal LDOI bases in any dimension.
\end{cor}

\begin{proof}
  Since $\max\left\{ |a_{i,j}|^2, |a_{j,i}|^2 \right\} \geq 1/2$ for all $i
  \neq j$ and $\max_{\sigma \in S_n} \sum_{i=1}^n |u_{i,\sigma(i)}|^2 \geq 1$
  (because the matrix with entries $|u_{i,j}|^2$ is doubly stochastic), the
  lower bound~\eqref{eq:ldoi_lb} immediately gives
  \begin{equation}
    \opt_{\locc}(\E) \geq \frac{1}{n^2} \left( \frac{n(n-1)}{2} + 1 \right) = \frac{1}{n^2} \frac{n^2 - n + 2}{2} = \frac{n^2 - n + 2}{2n^2} = \frac{1}{2} - \frac{n - 2}{2n^2}.
  \end{equation}
  To find the dimension-independent lower bound, we minimize $f(n) = 1/2 - (n-2)/(2n^2)$ for $n \geq 2$. Taking the derivative,
  \begin{equation}
    f'(n) = \frac{n - 4}{2n^3},
  \end{equation}
  which is negative for $n < 4$ and positive for $n > 4$, so the minimum occurs at $n = 4$, giving $f(4) = 7/16$.
\end{proof}

The lower bound expression~\eqref{eq:ldoi_lb} of
Theorem~\ref{thm:ldoi_basis_one_copy} evaluates to the universal bound value
for specific bases. Consider an orthonormal LDOI basis where $U$ is the $n
\times n$ Fourier matrix (normalized so that all entries have magnitude
$1/\sqrt{n}$) and $A$ is defined by $a_{i,j} = 1/\sqrt{2}$ for all $i \neq j$
(with arbitrary diagonal entries). For this basis, the Fourier matrix has the
property that all entries have equal magnitude, making it maximally ``spread
out'' across rows and columns. This structure minimizes the optimal assignment
value: the maximum over permutations $\sigma \in S_n$ of $\sum_{i=1}^n
|u_{i,\sigma(i)}|^2$ equals exactly $n(1/n) = 1$, the minimum possible value
given unitarity. Combined with $\sum_{i \neq j} \max\{|a_{i,j}|^2, |a_{j,i}|^2\}
= n(n-1)(1/2)$, the lower bound~\eqref{eq:ldoi_lb} evaluates to precisely
$\frac{1}{2} - \frac{n-2}{2n^2}$, matching the universal bound. This shows that
the lower bound formula cannot be uniformly tightened.

While Corollary~\ref{cor:universal_bound} establishes
the rigorous lower bound of $7/16$ for all dimensions, it is natural to ask
whether a uniform bound of $1/2$ holds for all uniform ensembles over LDOI bases.
For $n = 2$, Corollary~\ref{cor:ppt_dist_2x2} shows that $\opt_{\locc}(\E)
\geq 1/2$ for all two-qubit LDOI uniform ensembles, with equality achieved by the
uniform ensemble over the Bell basis. For $n \geq 3$, the Fourier basis
example shows that the lower bound~\eqref{eq:ldoi_lb} evaluates to exactly
$\frac{1}{2} - \frac{n-2}{2n^2}$. Whether $\opt_{\locc}(\E)$ for the Fourier
basis equals this lower bound value or is strictly larger remains an open
question (see Example~\ref{exam:mixed_bell_basis}).

While equality $\opt_{\locc}(\E) = \opt_{\ppt}(\E)$ holds
for uniform ensembles over a broad class of LDOI bases---including all two-qubit
cases by Corollary~\ref{cor:ppt_dist_2x2} and those satisfying the conditions of
Corollary~\ref{cor:large_U}---it remains an open question whether this equality
holds for \emph{all} orthonormal LDOI bases. Example~\ref{exam:mixed_bell_basis}
exhibits a basis where the upper and lower bounds in
Theorem~\ref{thm:ldoi_basis_one_copy} do not match, suggesting the possibility
of a strict gap in higher dimensions. In any case, we can show that these two
quantities are always close together. In particular, the following result shows
that the maximal gap between $\opt_{\ppt}(\E)$ and
$\opt_{\locc}(\E)$ for uniform ensembles approaches $0$ as $n \rightarrow
\infty$, and the largest potential gap between these quantities in any dimension
is $1/16 = 0.0625$ (when $n = 4$).

\begin{cor}\label{cor:close_LDOI_LOCC_PPT}
  Let $n \geq 2$ be an integer and let $\X = \Y = \complex^n$. For any orthonormal
  LDOI basis of $\X \otimes \Y$, define the ensemble $\E = \{(1/n^2,
  \ketbra{\phi_{i,j}}{\phi_{i,j}}) : 1 \leq i,j \leq n\}$ with uniform prior
  distribution, where $\ket{\phi_{i,j}}$ are the basis vectors. Then
  \begin{equation}
    0 \leq \opt_{\ppt}(\E) - \opt_{\locc}(\E) \leq \frac{n - 2}{2n^2}.
  \end{equation}
\end{cor}

\begin{proof}
  The left inequality is trivial. To verify the right inequality, we subtract the bound~\eqref{eq:ldoi_lb} from~\eqref{eq:ppt_dist_weaker} to see that
  \begin{align}\label{eq:diff_locc_ppt_firstbound}
    \opt_{\ppt}(\E) - \opt_{\locc}(\E) \leq \frac{1}{n^2}\left( \sum_{i=1}^n \max\left\{ \frac{1}{2}, \max_{1 \leq k \leq n} \left\{ |u_{k,i}|^2 \right\} \right\} - \max_{\sigma \in S_n}\left\{ \sum_{i=1}^n |u_{i,\sigma(i)}|^2\right\} \right).
  \end{align}
  Suppose that there are $M$ columns of $U$ containing an entry with magnitude at least $1/\sqrt{2}$ ($M$ may equal $0$). By permuting the columns of $U$ we can assume without loss of generality that these are its first $M$ columns, and by permuting the rows of $U$ we can further assume without loss of generality that $|u_{i,i}| \geq 1/\sqrt{2}$ for $1 \leq i \leq M$. Then
  \begin{align}\label{eq:max_U_val}
    \sum_{i=1}^n \max\left\{ \frac{1}{2}, \max_{1 \leq k \leq n} \left\{ |u_{k,i}|^2 \right\} \right\} = \frac{n-M}{2} + \sum_{i=1}^M |u_{i,i}|^2.
  \end{align}
  Similarly, we claim that
  \begin{align}\label{eq:perm_U_lb}
    \max_{\sigma \in S_n}\left\{ \sum_{i=1}^n |u_{i,\sigma(i)}|^2\right\} & 
    \geq \sum_{i=1}^M |u_{i,i}|^2 + 
    \begin{cases}
      \frac{n - 3M/2}{n-M} & \text{if $M < n$}, \\ 
      0 & \text{if $M = n$}.
    \end{cases}
  \end{align}
  To verify this inequality, we apply some Frobenius norm estimates to $U$: if we partition $U$ as a block matrix
  \begin{equation}
    U = \begin{bmatrix}
      U_1 & U_2 \\ U_3 & U_4
    \end{bmatrix}
  \end{equation}
  with $U_1$ being $M \times M$, then $\|U_1\|_{\textup{F}}^2 \geq \sum_{i=1}^M
  |u_{i,i}|^2 \geq M/2$, since $|u_{i,i}|^2 \geq 1/2$ for all $1 \leq i \leq
  M$. Now use the fact that $U$ is unitary to see that $\|U_1\|_{\textup{F}}^2
  + \|U_2\|_{\textup{F}}^2 = M$, which implies $\|U_2\|_{\textup{F}}^2 \leq M -
  M/2 = M/2$. Finally, using unitarity of $U$ again shows that
  $\|U_2\|_{\textup{F}}^2 + \|U_4\|_{\textup{F}}^2 = n-M$, so
  $\|U_4\|_{\textup{F}}^2 \geq n-M - M/2 = n - 3M/2$. Since $U_4$ is $(n-M)
  \times (n-M)$ with Frobenius norm squared at least $n - 3M/2$, and the
  optimal assignment for $U_4$ selects one entry per row and column, this
  assignment must achieve total squared magnitude at least $(n - 3M/2)/(n-M)$
  on average. Therefore, there must exist some $\sigma \in S_n$ with $\sigma(i)
  = i$ for $1 \leq i \leq M$ such that $\sum_{i=M+1}^n |u_{i,\sigma(i)}|^2 \geq
  (n - 3M/2)/(n-M)$, giving the bound~\eqref{eq:perm_U_lb}.
  
  Substituting Equation~\eqref{eq:max_U_val} and the
  bound~\eqref{eq:perm_U_lb} into
  Inequality~\eqref{eq:diff_locc_ppt_firstbound} then shows
  that
  \begin{equation}
    \opt_{\ppt}(\E) - \opt_{\locc}(\E) \leq \frac{n - M}{2n^2} - \begin{cases}\frac{n - 3M/2}{n^2(n-M)} & \text{if $M < n$}, \\ 0 & \text{if $M = n$.}\end{cases}
  \end{equation}
  The above bound is maximized when $M = 0$ or $M = n-1$, in which cases it
  simplifies to
  \begin{equation}
    \opt_{\ppt}(\E) - \opt_{\locc}(\E) \leq \frac{n - 2}{2n^2},
  \end{equation}
  which completes the proof.
\end{proof}

Table~\ref{tab:gap_bound} shows the gap bound $(n-2)/(2n^2)$ for small dimensions,
illustrating that the maximum gap of $1/16 = 0.0625$ occurs at $n=4$ and decreases
for larger $n$, vanishing as $n \to \infty$.

\begin{table}[!htpb]
\centering
\begin{tabular}{c|c|c}
$n$ & $(n-2)/(2n^2)$ & Decimal \\
\hline
2 & 0 & 0 \\
3 & $1/18$ & $\approx 0.0556$ \\
4 & $1/16$ & $= 0.0625$ \\
5 & $3/50$ & $= 0.06$ \\
10 & $1/25$ & $= 0.04$ \\
20 & $9/400$ & $= 0.0225$ \\
$\infty$ & 0 & 0
\end{tabular}
\caption{Upper bound on the gap $\opt_{\ppt}(\E) - \opt_{\locc}(\E)$ for uniform
LDOI ensembles as a function of dimension $n$. The gap is maximized at $n=4$.}
\label{tab:gap_bound}
\end{table}

\begin{example}\label{exam:mixed_bell_basis}
  Let $n \geq 3$ and consider an orthonormal LDOI basis arising from Definition~\ref{defn:LDOI_basis}
  with $A = (1/\sqrt{2})\1_n$ (the all-ones matrix
  scaled by $1/\sqrt{2}$) and $U$ equal to the $n$-dimensional Fourier matrix.
  Define the uniform ensemble $\E = \{(1/n^2, \ketbra{\phi_{i,j}}{\phi_{i,j}}) :
  1 \leq i,j \leq n\}$, where $\ket{\phi_{i,j}}$ are the basis vectors.
  This basis consists of $n$ maximally entangled states $\ket{\phi_{i,i}}$ and
  $n(n-1)$ Schmidt-rank-$2$ states $\ket{\phi_{i,j}}$ with $i \neq j$ that are
  symmetric or antisymmetric Bell-like states with Schmidt coefficients
  $1/\sqrt{2}$, $1/\sqrt{2}$.

  Since the Fourier matrix has all entries with magnitude $1/\sqrt{n}$, we have
  $\max_{1 \leq k \leq n} |u_{k,i}|^2 = 1/n < 1/2$ for all $i$. The bounds of
  Theorem~\ref{thm:ldoi_basis_one_copy} therefore show that
  \begin{equation}\label{eq:mixed_bell_bounds}
    \opt_{\locc}(\E) \geq \frac{1}{2} - \frac{n - 2}{2n^2} \quad \text{and} \quad \opt_{\ppt}(\E) \leq \frac{1}{2}.
  \end{equation}
  The upper and lower bounds do not match, indicating a potential gap between
  LOCC and PPT distinguishability. We now show that the upper bound is tight:
  $\opt_{\ppt}(\E) = 1/2$.

  For each ordered pair $(i,j)$ with $i \neq j$, define
  \begin{equation}\label{eq:mixed_bell_measurements}
    P_{i,j} = \frac{1}{2}\big(\ketbra{ij}{ij} + \ketbra{ji}{ji}\big)
    + \frac{s_{i,j}}{2n-2}\big(\ketbra{ij}{ji} + \ketbra{ji}{ij}\big)
    + \frac{1}{2n-2}\big(\ketbra{ii}{ii} + \ketbra{jj}{jj}\big),
  \end{equation}
  where $s_{i,j} = 1$ when $i < j$ and $s_{i,j} = -1$ when $i > j$.
  In particular $s_{j,i} = -s_{i,j}$, which will ensure the cross terms cancel
  when the operators are summed. For the diagonal pairs we set $P_{i,i} = 0$
  for all $1 \leq i \leq n$.

  Each $P_{i,j}$ with $i \neq j$ acts only on the $4$-dimensional subspace
  spanned by $\{\ket{ii}, \ket{ij}, \ket{ji}, \ket{jj}\}$. The block on
  $\{\ket{ij},\ket{ji}\}$ has matrix representation
  $\begin{psmallmatrix}1/2 & s_{i,j}/(2n-2) \\ s_{i,j}/(2n-2) & 1/2\end{psmallmatrix}$, whose eigenvalues are $1/2 \pm 1/(2n-2) \geq 0$, so it is positive
  semidefinite. The block on $\{\ket{ii},\ket{jj}\}$ is diagonal with positive
  entries $1/(2n-2)$. Taking the partial transpose swaps
  $\ket{ij}\!\bra{ji}$ with $\ket{jj}\!\bra{ii}$, which simply replaces the
  latter block with $\begin{psmallmatrix}1/(2n-2) & s_{i,j}/(2n-2) \\
  s_{i,j}/(2n-2) & 1/(2n-2)\end{psmallmatrix}$, again positive semidefinite. Hence every
  $P_{i,j}$ is PPT.

  Summing over all ordered pairs gives
  \begin{equation}
    \sum_{i,j=1}^n P_{i,j} = 
    \sum_{\substack{i,j=1 \\ i \neq j}}^n P_{i,j}
    = \sum_{\substack{i,j = 1 \\ i \neq j}}^n \frac{1}{2}(\ketbra{ij}{ij} + \ketbra{ji}{ji})
      + \sum_{\substack{i,j=1 \\ i \neq j}}^n \frac{1}{2n-2}(\ketbra{ii}{ii} + \ketbra{jj}{jj}),
  \end{equation}
  because the off-diagonal terms $\ketbra{ij}{ji}$ cancel between
  $P_{i,j}$ and $P_{j,i}$ (which have opposite signs). Every basis vector
  $\ket{ij}$ with $i \neq j$ appears in exactly two summands, each contributing
  $1/2$, and every $\ket{ii}$ appears in $2(n-1)$ summands, each contributing
  $1/(2n-2)$. Thus $\sum_{i,j} P_{i,j} = \I$.

  A direct calculation shows
  \begin{equation}
    \ip{P_{i,j}}{\phi_{i,j}} = \frac{1}{2} + \frac{1}{2n-2} = \frac{n}{2n-2}
  \end{equation}
  for all $i \neq j$, regardless of whether $i < j$ or $i > j$: the symmetric
  (resp., antisymmetric) combinations pick up the positive (resp., negative)
  cross term with the same magnitude. Because $\ip{P_{i,i}}{\phi_{i,i}} = 0$
  for all $i$, the average success probability is
  \begin{equation}
    \frac{1}{n^2}\sum_{i,j=1}^n \ip{P_{i,j}}{\phi_{i,j}}
    = \frac{1}{n^2} \cdot n(n-1) \cdot \frac{n}{2n-2} = \frac{1}{2},
  \end{equation}
  establishing $\opt_{\ppt}(\E) = 1/2$.

  Each $P_{i,j}$ is supported on a $(2 \otimes 2)$-dimensional subsystem
  (the subspace spanned by $\{\ket{ii}, \ket{ij}, \ket{ji}, \ket{jj}\}
  \cong \complex^2 \otimes \complex^2$). Since PPT is equivalent to
  separability for $2 \otimes 2$ systems, the measurements are also
  separable. Thus $\opt_{\sep}(\E) = 1/2$ as well.

  Whether the lower bound in Equation~\eqref{eq:mixed_bell_bounds} is
  tight---that is, whether $\opt_{\locc}(\E) = 1/2 -
  (n-2)/(2n^2)$---remains an open question. Whether $\opt_{\locc}(\E)$ coincides with one of these bounds or takes an
  intermediate value remains unknown; resolving this would require new techniques
  for analyzing the full LOCC hierarchy beyond the product measurement strategy
  of Theorem~\ref{thm:ldoi_basis_one_copy}. 
\end{example}

Corollaries~\ref{cor:large_U} and~\ref{cor:ppt_dist_2x2}, as well as Example~\ref{exam:mixed_bell_basis}, might lead us to conjecture that, for every LDOI basis with uniform ensemble $\E$, the upper bound~\eqref{eq:ldoi_ub} on $\opt_{\ppt}(\E)$ is actually equality (for some suitably chosen $c_1$, $\ldots$, $c_n$). We now present the simplest example that we have been able to find to demonstrate that this is not the case.

\begin{example}\label{exam:PPT_bound_not_attained}
  Let $n = 3$ and $\X = \Y = \complex^n$. Consider an LDOI basis of $\X \otimes \Y$
  arising from Definition~\ref{defn:LDOI_basis} via the matrices
  \begin{equation}
    A = \frac{1}{5}\begin{bmatrix}
      0 & 3 & 3 \\
      4 & 0 & 3 \\
      4 & 4 & 0
    \end{bmatrix} \quad \text{and} \quad U = \frac{1}{3}\begin{bmatrix}
      2 & -2 & 1 \\
      1 & 2 & 2 \\
      2 & 1 & -2
    \end{bmatrix}.
  \end{equation}
  Define the uniform ensemble $\E = \{(1/9, \ketbra{\phi_{i,j}}{\phi_{i,j}}) : 1 \leq
  i,j \leq 3\}$, where $\ket{\phi_{i,j}}$ are the basis vectors.
  Semidefinite programming shows that $c_1 + c_2 + c_3$ is minimized (subject to the constraints~\eqref{eq:ldoi_c_bound}) when $c_1 = c_2 = c_3 = 12/25$. Note that in this example, the constraint $c_ic_j \geq |a_{i,j}|^2|a_{j,i}|^2$ is satisfied with equality for all $i \neq j$: we have $c_ic_j = (12/25)^2 = 144/625 = (3/5)^2(4/5)^2 = |a_{i,j}|^2|a_{j,i}|^2$ exactly. With these values, the best bound~\eqref{eq:ldoi_ub} is equal to
  \begin{equation}
    \frac{1}{n^2}\left(\sum_{\substack{i,j=1 \\ i \neq j}}^n \max\left\{ |a_{i,j}|^2, |a_{j,i}|^2 \right\} + \sum_{i=1}^n c_i\right) = 44/75 \approx 0.5867.
  \end{equation}
  However, solving the semidefinite program~\eqref{eq:ppt-primal} directly for this ensemble $\E$ shows that
  \begin{equation}
    \opt_{\ppt}(\E) = 26/45 \approx 0.5778,
  \end{equation}
  demonstrating that the upper bound~\eqref{eq:ldoi_ub} on $\opt_{\ppt}(\E)$ is not always attained. The gap arises because the upper bound~\eqref{eq:ldoi_ub} is derived from
  the constrained dual~\eqref{eq:ppt-dual-constrained} with $H$ restricted to be
  diagonal. For this particular $(U,A)$ pair, the columns of $U$ have no entry
  with magnitude at least $1/\sqrt{2}$, so the condition of
  Corollary~\ref{cor:large_U} fails. In the full dual~\eqref{eq:ppt-dual}, the
  auxiliary variables $Q_k$ provide additional degrees of freedom that allow a
  tighter certificate; restricting to the diagonal ansatz $H$ forfeits this
  flexibility. More generally, the bound~\eqref{eq:ldoi_ub} is tight precisely
  when the diagonal ansatz suffices to certify optimality, which holds whenever
  the column-magnitude condition of Corollary~\ref{cor:large_U} is satisfied.
\end{example}

\section{Conclusion} 
\label{sec:conclusion}

We have studied the local distinguishability of quantum states with local
diagonal orthogonal invariance---a broad class that includes Werner states,
isotropic states, X-states, and Dicke states. Our main structural result,
Theorem~\ref{thm:ldoi_sdp_simplifies}, shows that for LDOI ensembles, the
search for optimal measurements can be restricted to the LDOI subspace without
loss of generality. For orthonormal LDOI bases, we established efficiently computable
upper and lower bounds on distinguishability (Theorem~\ref{thm:ldoi_basis_one_copy})
and proved that the LOCC supremum equals the PPT and separable optima for a
broad class of bases, including all two-qubit cases
(Corollary~\ref{cor:ppt_dist_2x2}) and bases with sufficiently large entries in
the unitary matrix $U$ (Corollary~\ref{cor:large_U}). More generally, we showed
that the gap between PPT and LOCC distinguishability is bounded by $(n-2)/(2n^2)$,
which achieves its maximum value of $1/16$ at $n=4$ and vanishes asymptotically.
The LOCC lower bound evaluates to exactly $1/2 - (n-2)/(2n^2)$ for the Fourier
basis with uniform ensemble $\E$. Whether $\opt_{\locc}(\E)$ equals this lower
bound value or is strictly larger remains open.

Several natural questions remain open and suggest promising directions for
future investigation. First, while we have shown that $\opt_{\locc}(\E)
= \opt_{\ppt}(\E)$ for a broad class of uniform LDOI ensembles
(Corollaries~\ref{cor:large_U} and~\ref{cor:ppt_dist_2x2}), it remains open
whether this equality holds for \emph{all} uniform ensembles of orthonormal LDOI bases.
Example~\ref{exam:mixed_bell_basis} demonstrates that the bounds in
Theorem~\ref{thm:ldoi_basis_one_copy} do not always coincide, but this does not
rule out equality between the optimal values. It would also be interesting
to extend the explicit bounds of Theorem~\ref{thm:ldoi_basis_one_copy} to
non-uniform prior distributions and to understand how the relationship between
LOCC and PPT distinguishability depends on the prior.

Second, orthonormal LDOI bases are typically not perfectly distinguishable
by PPT measurements (Example~\ref{exam:mixed_bell_basis}), so it is natural to
ask whether they can at least be \emph{unambiguously} distinguished by PPT
measurements: that is, whether there is a PPT POVM $\{P_{i,j}\} \cup \{P_?\}$,
with $P_?$ the inconclusive outcome, that never misidentifies a state and
identifies each state conclusively with non-zero probability. We now show that the answer is negative unless the basis consists entirely of product states:

\begin{prop}\label{prop:unambiguous_ppt}
  An orthonormal LDOI basis admits a PPT POVM that conclusively identifies each of its states with non-zero probability if and only if it consists entirely of product states, i.e., if and only if every entry of the matrices $A$ and $U$ of Definition~\ref{defn:LDOI_basis} has magnitude $0$ or $1$.
\end{prop}

\begin{proof}
  Recall from \cite{Che00} that if $\{\ket{\phi_1}, \ldots, \ket{\phi_{n^2}}\}$ is a (not necessarily LDOI) orthonormal basis, then every POVM $\{P_1, \ldots, P_{n^2}, P_?\}$ with $\ip{P_k}{\phi_\ell} = 0$ whenever $k \neq \ell$ satisfies $P_k = c_k\ketbra{\phi_k}{\phi_k}$ for some $c_k \in [0,1]$ (for all $1 \leq k \leq n^2$). Since $P_k$ has rank at most $1$, $P_k \in \ppt(\X : \Y)$ implies that $P_k$ is a product vector, so $c_k = 0$ or $\ketbra{\phi_k}{\phi_k}$ is a product state.

  For the final claim about the matrices $A$ and $U$, recall from Definition~\ref{defn:LDOI_basis} that
  $\ket{\phi_{i,i}} = \sum_k u_{i,k}\ket{k} \otimes \ket{k}$ is a product state
  exactly when a single entry of the $i$-th row of $U$ is non-zero, and that
  $\ket{\phi_{i,j}}$ and $\ket{\phi_{j,i}}$ are product states exactly when
  $a_{i,j} = 0$ or $a_{j,i} = 0$.
\end{proof}

Proposition~\ref{prop:unambiguous_ppt} shows that unambiguous
discrimination of a \emph{complete} orthonormal basis is a degenerate problem:
the unambiguity constraint alone pins each POVM element down to a multiple of
the state it identifies, leaving the inconclusive outcome $P_? = \I$ as the only
PPT strategy whenever some basis state is entangled. The question becomes
non-trivial for proper subsets of an orthonormal LDOI basis, whose elements no
longer span $\X \otimes \Y$ and for which the unambiguity constraint leaves
genuine freedom in the choice of the $P_k$. Determining the optimal unambiguous
PPT success probability for such subsets, and comparing it against the
corresponding LOCC and separable quantities, is left as an open problem.

\subsection{Software}
\label{sec:software}
Companion software for determining whether a collection of quantum states
constitute an LDOI basis, implementing the LDOT map, along with the semidefinite
programs for determining the PPT distinguishability of a collection of quantum
states in this work are provided within the \texttt{toqito} quantum information
package~\cite{russo2021toqito} which leverages
\texttt{PICOS}~\cite{sagnol2022picos} and the \texttt{CVXOPT}
solver~\cite{andersen2020cvxopt} for semidefinite program optimization.
The source code is openly available at
\url{https://github.com/vprusso/toqito}, from which all of the numerical
results reported here can be reproduced.

\subsection*{Acknowledgments}

N.J.\ was supported by NSERC Discovery Grant RGPIN-2022-04098.

\subsection*{Author contributions}

Both authors contributed equally to this work.

\bibliographystyle{quantum}
\bibliography{refs}

\begin{thebibliography}{10}

\bibitem{bandyopadhyay2011more}
Somshubhro Bandyopadhyay.
\newblock ``More nonlocality with less purity''.
\newblock \href{https://dx.doi.org/10.1103/PhysRevLett.106.210402}{Physical
  Review Letters {\bf 106}, 210402}~(2011).

\bibitem{bennett1999quantum}
Charles~H Bennett, David~P DiVincenzo, Christopher~A Fuchs, Tal Mor, Eric
  Rains, Peter~W Shor, John~A Smolin, and William~K Wootters.
\newblock ``Quantum nonlocality without entanglement''.
\newblock \href{https://dx.doi.org/10.1103/PhysRevA.59.1070}{Physical Review A
  {\bf 59}, 1070}~(1999).

\bibitem{ghosh2002local}
Sibasish Ghosh, Guruprasad Kar, Anirban Roy, Debasis Sarkar, Aditi Sen, Ujjwal
  Sen, et~al.
\newblock ``Local indistinguishability of orthogonal pure states by using a
  bound on distillable entanglement''.
\newblock \href{https://dx.doi.org/10.1103/PhysRevA.65.062307}{Physical Review
  A {\bf 65}, 062307}~(2002).

\bibitem{horodecki2003local}
Micha{\l} Horodecki, Aditi Sen, Ujjwal Sen, and Karol Horodecki.
\newblock ``Local indistinguishability: {M}ore nonlocality with less
  entanglement''.
\newblock \href{https://dx.doi.org/10.1103/PhysRevLett.90.047902}{Physical
  Review Letters {\bf 90}, 047902}~(2003).

\bibitem{virmani2001optimal}
Shashank Virmani, Massimiliano~F Sacchi, Martin~B Plenio, and Damian Markham.
\newblock ``Optimal local discrimination of two multipartite pure states''.
\newblock \href{https://dx.doi.org/10.1016/S0375-9601(01)00484-4}{Physics
  Letters A {\bf 288}, 62--68}~(2001).

\bibitem{walgate2002nonlocality}
Jonathan Walgate and Lucien Hardy.
\newblock ``Nonlocality, asymmetry, and distinguishing bipartite states''.
\newblock \href{https://dx.doi.org/10.1103/PhysRevLett.89.147901}{Physical
  Review Letters {\bf 89}, 147901}~(2002).

\bibitem{walgate2000local}
Jonathan Walgate, Anthony~J Short, Lucien Hardy, and Vlatko Vedral.
\newblock ``Local distinguishability of multipartite orthogonal quantum
  states''.
\newblock \href{https://dx.doi.org/10.1103/PhysRevLett.85.4972}{Physical Review
  Letters {\bf 85}, 4972}~(2000).

\bibitem{helstrom1969quantum}
Carl~W. Helstrom.
\newblock ``Quantum detection and estimation theory''.
\newblock \href{https://dx.doi.org/10.1007/BF01007479}{Journal of Statistical
  Physics {\bf 1}, 231--252}~(1969).

\bibitem{holevo1973statistical}
Alexander~S. Holevo.
\newblock ``Statistical decision theory for quantum systems''.
\newblock \href{https://dx.doi.org/10.1016/0047-259X(73)90028-6}{Journal of
  Multivariate Analysis {\bf 3}, 337--394}~(1973).

\bibitem{yuen1975optimum}
Horace~P. Yuen, Robert~S. Kennedy, and Melvin Lax.
\newblock ``Optimum testing of multiple hypotheses in quantum detection
  theory''.
\newblock \href{https://dx.doi.org/10.1109/TIT.1975.1055351}{IEEE Transactions
  on Information Theory {\bf 21}, 125--134}~(1975).

\bibitem{chitambar2014everything}
Eric Chitambar, Debbie Leung, Laura Man{\v{c}}inska, Maris Ozols, and Andreas
  Winter.
\newblock ``Everything you always wanted to know about {LOCC} (but were afraid
  to ask)''.
\newblock \href{https://dx.doi.org/10.1007/s00220-014-1953-9}{Communications in
  Mathematical Physics {\bf 328}, 303--326}~(2014).

\bibitem{duan2009distinguishability}
Runyao Duan, Yuan Feng, Yu~Xin, and Mingsheng Ying.
\newblock ``Distinguishability of quantum states by separable operations''.
\newblock \href{https://dx.doi.org/10.1109/TIT.2008.2011524}{IEEE Transactions
  on Information Theory {\bf 55}, 1320--1330}~(2009).

\bibitem{bandyopadhyay2015limitations}
Somshubhro Bandyopadhyay, Alessandro Cosentino, Nathaniel Johnston, Vincent
  Russo, John Watrous, and Nengkun Yu.
\newblock ``Limitations on separable measurements by convex optimization''.
\newblock \href{https://dx.doi.org/10.1109/TIT.2015.2417755}{IEEE Transactions
  on Information Theory {\bf 61}, 3593--3604}~(2015).

\bibitem{cosentino2013positive}
Alessandro Cosentino.
\newblock ``Positive-partial-transpose-indistinguishable states via
  semidefinite programming''.
\newblock \href{https://dx.doi.org/10.1103/PhysRevA.87.012321}{Physical Review
  A {\bf 87}, 012321}~(2013).

\bibitem{yu2014distinguishability}
Nengkun Yu, Runyao Duan, and Mingsheng Ying.
\newblock ``Distinguishability of quantum states by positive operator-valued
  measures with positive partial transpose''.
\newblock \href{https://dx.doi.org/10.1109/TIT.2014.2307575}{IEEE Transactions
  on Information Theory {\bf 60}, 2069--2079}~(2014).

\bibitem{yu2012four}
Nengkun Yu, Runyao Duan, and Mingsheng Ying.
\newblock ``Four locally indistinguishable ququad-ququad orthogonal maximally
  entangled states''.
\newblock \href{https://dx.doi.org/10.1103/PhysRevLett.109.020506}{Physical
  Review Letters {\bf 109}, 020506}~(2012).

\bibitem{chruscinski2006class}
Dariusz Chru{\'s}ci{\'n}ski and Andrzej Kossakowski.
\newblock ``Class of positive partial transposition states''.
\newblock \href{https://dx.doi.org/10.1103/PhysRevA.74.022308}{Physical Review
  A—Atomic, Molecular, and Optical Physics {\bf 74}, 022308}~(2006).

\bibitem{johnston2019pairwise}
Nathaniel Johnston and Olivia MacLean.
\newblock ``Pairwise completely positive matrices and conjugate local diagonal
  unitary invariant quantum states''.
\newblock \href{https://dx.doi.org/10.13001/1081-3810.3842}{Electronic Journal
  of Linear Algebra {\bf 35}, 156--180}~(2019).

\bibitem{nechita2021graphical}
Ion Nechita and Satvik Singh.
\newblock ``A graphical calculus for integration over random diagonal unitary
  matrices''.
\newblock \href{https://dx.doi.org/10.1016/j.laa.2020.12.014}{Linear Algebra
  and its Applications {\bf 613}, 46--86}~(2021).

\bibitem{nechita2021diagonal}
Satvik Singh and Ion Nechita.
\newblock ``Diagonal unitary and orthogonal symmetries in quantum theory''.
\newblock \href{https://dx.doi.org/10.22331/q-2021-08-09-519}{Quantum {\bf 5},
  519}~(2021).

\bibitem{werner1989quantum}
Reinhard~F. Werner.
\newblock ``Quantum states with {E}instein-{P}odolsky-{R}osen correlations
  admitting a hidden-variable model''.
\newblock \href{https://dx.doi.org/10.1103/PhysRevA.40.4277}{Physical Review A
  {\bf 40}, 4277--4281}~(1989).

\bibitem{horodecki1999reduction}
Micha{\l} Horodecki and Pawel Horodecki.
\newblock ``Reduction criterion of separability and limits for a class of
  distillation protocols''.
\newblock \href{https://dx.doi.org/10.1103/PhysRevA.59.4206}{Physical Review A
  {\bf 59}, 4206--4216}~(1999).

\bibitem{quesada2012quantum}
Nicolás Quesada, Asma Al-Qasimi, and Daniel F.~V. James.
\newblock ``Quantum properties and dynamics of {X} states''.
\newblock \href{https://dx.doi.org/10.1080/09500340.2012.713130}{Journal of
  Modern Optics {\bf 59}, 1322--1329}~(2012).

\bibitem{yu2016separability}
Nengkun Yu.
\newblock ``Separability of a mixture of {D}icke states''.
\newblock \href{https://dx.doi.org/10.1103/PhysRevA.94.060101}{Physical Review
  A {\bf 94}, 060101(R)}~(2016).

\bibitem{tura2018separability}
Jordi Tura, Albert Aloy, Ruben Quesada, Maciej Lewenstein, and Anna Sanpera.
\newblock ``Separability of diagonal symmetric states: {A} quadratic conic
  optimization problem''.
\newblock \href{https://dx.doi.org/10.22331/q-2018-01-12-45}{Quantum {\bf 2},
  45}~(2018).

\bibitem{ha2011one}
Kil-Chan Ha and Seung-Hyeok Kye.
\newblock ``One-parameter family of indecomposable optimal entanglement
  witnesses arising from generalized {C}hoi maps''.
\newblock \href{https://dx.doi.org/10.1103/PhysRevA.84.024302}{Physical Review
  A {\bf 84}, 024302}~(2011).

\bibitem{bennett1992communication}
Charles~H. Bennett and Stephen~J. Wiesner.
\newblock ``Communication via one- and two-particle operators on
  {E}instein-{P}odolsky-{R}osen states''.
\newblock \href{https://dx.doi.org/10.1103/PhysRevLett.69.2881}{Physical Review
  Letters {\bf 69}, 2881--2884}~(1992).

\bibitem{ambainis1999dense}
Andris Ambainis, Ashwin Nayak, Amnon Ta-Shma, and Umesh Vazirani.
\newblock ``Dense quantum coding and a lower bound for 1-way quantum
  automata''.
\newblock In Proceedings of the thirty-first annual ACM symposium on Theory of
  computing.
\newblock \href{https://dx.doi.org/10.1145/301250.301347}{Pages 376--383}.
\newblock ACM~(1999).

\bibitem{kuhn1955hungarian}
Harold~W. Kuhn.
\newblock ``The {H}ungarian method for the assignment problem''.
\newblock \href{https://dx.doi.org/10.1002/nav.3800020109}{Naval Research
  Logistics Quarterly {\bf 2}, 83--97}~(1955).

\bibitem{munkres1957algorithms}
James Munkres.
\newblock ``Algorithms for the assignment and transportation problems''.
\newblock \href{https://dx.doi.org/10.1137/0105003}{Journal of the Society for
  Industrial and Applied Mathematics {\bf 5}, 32--38}~(1957).

\bibitem{tomizawa1971some}
Nobuaki Tomizawa.
\newblock ``On some techniques useful for solution of transportation network
  problems''.
\newblock \href{https://dx.doi.org/10.1002/net.3230010206}{Networks {\bf 1},
  173--194}~(1971).

\bibitem{edmonds1972theoretical}
Jack Edmonds and Richard~M. Karp.
\newblock ``Theoretical improvements in algorithmic efficiency for network flow
  problems''.
\newblock \href{https://dx.doi.org/10.1145/321694.321699}{Journal of the ACM
  {\bf 19}, 248--264}~(1972).

\bibitem{bandyopadhyay2013tight}
Somshubhro Bandyopadhyay and Michael Nathanson.
\newblock ``Tight bounds on the distinguishability of quantum states under
  separable measurements''.
\newblock \href{https://dx.doi.org/10.1103/PhysRevA.88.052313}{Physical Review
  A {\bf 88}, 052313}~(2013).

\bibitem{nielsen2002quantum}
Michael~A. Nielsen and Isaac~L. Chuang.
\newblock ``Quantum {C}omputation and {Q}uantum {I}nformation''.
\newblock \href{https://dx.doi.org/10.1017/CBO9780511976667}{Cambridge
  University Press}. ~(2010).

\bibitem{wilde2013quantum}
Mark~M. Wilde.
\newblock ``Quantum {I}nformation {T}heory''.
\newblock \href{https://dx.doi.org/10.1017/9781316809976}{Cambridge University
  Press}. ~(2013).

\bibitem{watrous2018theory}
John Watrous.
\newblock ``The {T}heory of {Q}uantum {I}nformation''.
\newblock \href{https://dx.doi.org/10.1017/9781316848142}{Cambridge University
  Press}. ~(2018).

\bibitem{cosentino2013small}
Alessandro Cosentino and Vincent Russo.
\newblock ``Small sets of locally indistinguishable orthogonal maximally
  entangled states''.
\newblock \href{https://dx.doi.org/10.48550/arXiv.1307.3232}{Quantum
  Information and Computation {\bf 14}, 1098--1106}~(2014).

\bibitem{Che00}
Anthony Chefles.
\newblock ``Quantum state discrimination''.
\newblock \href{https://dx.doi.org/10.1080/00107510010002599}{Contemporary
  Physics {\bf 41}, 401--424}~(2000).

\bibitem{russo2021toqito}
Vincent Russo.
\newblock ``toqito -- {Theory} of quantum information toolkit: {A} {Python}
  package for studying quantum information''.
\newblock \href{https://dx.doi.org/10.21105/joss.03082}{Journal of Open Source
  Software {\bf 6}, 3082}~(2021).

\bibitem{sagnol2022picos}
Guillaume Sagnol and Maximilian Stahlberg.
\newblock ``{PICOS}: {A} {Python} interface to conic optimization solvers''.
\newblock \href{https://dx.doi.org/10.21105/joss.03915}{Journal of Open Source
  Software {\bf 7}, 3915}~(2022).

\bibitem{andersen2020cvxopt}
Martin Andersen, Joachim Dahl, and Lieven Vandenberghe.
\newblock ``{CVXOPT}: {Convex} {Optimization}''.
\newblock Astrophysics Source Code Library, record
  \href{https://ui.adsabs.harvard.edu/abs/2020ascl.soft08017A}{ascl:2008.017}~(2020).

\end{thebibliography}

\end{document}